\documentclass{lmcs}
\pdfoutput=1

\usepackage{lastpage}
\lmcsdoi{16}{1}{20}
\lmcsheading{}{\pageref{LastPage}}{}{}%
{Oct.~12,~2018}{Feb.~18,~2020}{}

\usepackage{microtype, hyperref, mathpartir, xspace, xparse, stmaryrd, amssymb}
\usepackage[utf8]{inputenc}

\newcommand{\Nat}{\mathbb N}

\newcommand{\restr}{{\restriction}}
\renewcommand{\phi}{\varphi}
\renewcommand{\epsilon}{\varepsilon}
\newcommand{\Pht}[1]{\mathit{Pht}_{#1}}
\newcommand{\pht}[3]{[#2]_{#1}^{#3}}
\newcommand{\Comp}{\mathit{Comp}}
\newcommand{\arr}{\mathbin{\to}}
\newcommand{\BT}{\mathit{BT}}

\newcommand{\unbound}{\ensuremath{\mathsf U}\xspace}
\newcommand{\msoufin}{{\upshape MSO+}\ensuremath{\mathsf{U^{fin}}}\xspace}
\newcommand{\wmsou}{{\upshape WMSO+}\unbound}
\newcommand{\msou}{{\upshape MSO+}\unbound}
\newcommand{\mso}{{\upshape MSO}\xspace}
\newcommand{\child}{\curlywedgedownarrow}
\newcommand{\Aa}{{\mathcal A}}
\newcommand{\Bb}{{\mathcal B}}
\newcommand{\Cc}{{\mathcal C}}
\newcommand{\Uu}{{\mathcal U}}

\newcommand{\Gg}{{\mathcal G}}
\newcommand{\Ll}{{\mathcal L}}
\newcommand{\Nn}{{\mathcal N}}
\newcommand{\Pp}{{\mathcal P}}
\newcommand{\Ppfin}{\Pp^\mathsf{fin}}

\newcommand{\Rr}{{\mathcal R}}
\newcommand{\Tt}{{\mathcal T}}
\newcommand{\Vv}{{\mathcal V}}
\newcommand{\Xvar}{\mathsf{X}}
\newcommand{\Yvar}{\mathsf{Y}}
\newcommand{\Zvar}{\mathsf{Z}}
\newcommand{\Fvar}{\mathsf{F}}
\newcommand{\Vars}{\mathit{Vars}}
\newcommand{\Mnon}{M}%\mathsf{M}}
\newcommand{\Nnon}{N}%\mathsf{N}}
\newcommand{\efin}{{\exists_{\mathsf{fin}}}}
\renewcommand{\nd}{\mathsf{nd}}
\newcommand{\dom}{\mathrm{dom}}

\newcommand{\SUP}{\mathsf{SUP}}
\newcommand{\rednd}{\to_\nd}

\newcommand{\args}{s}
\newcommand{\nChld}{r}
\newcommand{\rMax}{\nChld_{\max}}
\newcommand{\tK}{K}
\newcommand{\tL}{L}
\newcommand{\xVar}{\mathsf{x}}
\newcommand{\unode}{u}
\newcommand{\vnode}{v}
\newcommand{\wnode}{w}
\newcommand{\VALempty}{\nu_\emptyset}

\newcommand{\true}{\mathsf{tt}}
\newcommand{\false}{\mathsf{ff}}
\newcommand{\Sigmaout}{\Sigma^\mathsf{out}}
\newcommand{\Sigmain}{\Sigma^\mathsf{in}}
\newcommand{\NIni}{N_\mathsf{0}}
\newcommand{\qIni}{q_\mathsf{0}}
\newcommand{\pIni}{p_\mathsf{0}}
\newcommand{\QImp}{Q_\mathsf{imp}}
\newcommand{\SigmaF}{\Sigma_\mathsf{F}}
\newcommand{\otyp}{\mathsf{o}}

\newcommand{\PD}[2]{\mathcal{PD}_{#1,#2}}
\newcommand{\op}{\mathit{op}}
\newcommand{\id}{\mathsf{id}}
\newcommand{\Op}[2]{\mathit{Op}_{#1,#2}}
\newcommand{\topS}{\mathit{top}}
\newcommand{\lamdots}{.\cdots{}.}

\NewDocumentCommand{\symb}{ m o o o o o o }{#1\langle\IfValueT{#2}{#2}\IfValueT{#3}{,#3}\IfValueT{#4}{,#4}\IfValueT{#5}{,#5}\IfValueT{#6}{,#6}\IfValueT{#7}{,#7}\rangle}

\definecolor{darkgreen}{RGB}{0,191,0}

\newcommand{\newqed}{\qed}

\theoremstyle{definition}
\newtheorem{exam}[thm]{Example}
\newtheorem*{rem*}{Remark}

\theoremstyle{plain}
\newtheorem{factfact}[thm]{Fact}
\theoremstyle{thmC} % without [ ]-brackets
\newtheorem{factfactC}[thm]{Fact}


\begin{document}

\title{Recursion Schemes, the \mso Logic, and the \unbound quantifier}

\author[P.~Parys]{Paweł Parys}
\address{Institute of Informatics, University of Warsaw}
\email{parys@mimuw.edu.pl}
\thanks{Work supported by the National Science Centre, Poland (grant no.\ 2016/22/E/ST6/00041).}

\begin{abstract}
	We study the model-checking problem for recursion schemes: does the tree generated by a given higher-order recursion scheme satisfy a given logical sentence.
	The problem is known to be decidable for sentences of the \mso logic.
	We prove decidability for an extension of \mso in which we additionally have an unbounding quantifier \unbound,
	saying that a subformula is true for arbitrarily large finite sets.
	This quantifier can be used only for subformulae in which all free variables represent finite sets
	(while an unrestricted use of the quantifier leads to undecidability).

	We also show that the logic has the properties of reflection and effective selection for trees generated by recursion schemes.
\end{abstract}

\maketitle

\section{Introduction}\label{sec:intro}

	\emph{Higher-order recursion schemes} (\emph{schemes} in short) are used to faithfully represent the control flow of programs in languages with higher-order functions~\cite{Damm82, KNU-hopda, Ong-hoschemes, KobayashiPrograms}.
	This formalism is equivalent via direct translations to simply-typed $\lambda Y$-calcu\-lus~\cite{schemes-lY}.
	Collapsible pushdown systems~\cite{collapsible} and ordered tree-pushdown systems~\cite{Ordered-Tree-Pushdown} are other equivalent formalisms.
	Schemes cover some other models such as indexed grammars~\cite{AhoIndexed} and ordered multi-pushdown automata~\cite{OrderedMultiPushdown}.

	In our setting, a scheme is a finite description of an infinite tree.
	A useful property of schemes is that the \emph{MSO-model-checking problem} for schemes is decidable.
	This means that given a scheme $\Gg$ and an \mso sentence $\phi$, it can be algorithmically decided whether the tree generated by $\Gg$ satisfies $\phi$.
	This result has several different proofs~\cite{Ong-hoschemes, collapsible, KobayashiOngtypes, KrivineWS},
	and also some extensions like global model checking~\cite{globalMC}, logical reflection~\cite{reflection}, effective selection~\cite{selection}, existence of lambda-calculus model~\cite{ModelSM}.
	When the property of trees is given as an automaton, not as a formula, the model-checking problem can be solved efficiently,
	in the sense that there exist implementations working in a reasonable running time~\cite{KobayashiPrograms, GTRecS, HorSat, PrefaceTool, TravMC2} (most tools cover only a fragment of \mso, though).

	Recently, an interest has arisen in model-checking trees generated by schemes against properties not expressible in the \mso logic.
	These are properties expressing boundedness and unboundedness of some quantities.
	More precisely, it was shown that the \emph{simultaneous unboundedness problem} (aka.~\emph{diagonal problem}) for schemes is decidable~\cite{diagonal-safe, diagonal, types-diagonal}.
	This problem asks, given a scheme $\Gg$ and a set of letters $A$, whether for every $n\in\Nat$ there exists a path in the tree generated by $\Gg$ such that every letter from $A$ appears on this path at least $n$ times.
	This result turns out to be interesting, because it entails other decidability results for recursion schemes,
	concerning in particular computability of the downward closure of recognized languages~\cite{Zetzsche-dc},
	and the problem of separability by piecewise testable languages~\cite{sep-piecewise-test}.

	In this paper we show a result of a more general style.
	Instead of considering a particular property, like in the simultaneous unboundedness problem, we consider a logic capable to express properties talking about boundedness.
	More precisely, we extend the \mso logic by the unbounding quantifier, \unbound~\cite{BojanczykU}.
	A formula using this quantifier, $\unbound \Xvar.\phi$, says that $\phi$ holds for arbitrarily large finite sets $\Xvar$.
	We impose a restriction that $\unbound \Xvar.\phi$ can be used only in a context where all free variables of $\phi$ represent finite sets.
	We call the resulting logic \msoufin.

	The goal of this paper is to prove the following theorem.

	\begin{thm}\label{thm:main}
		It is decidable whether the tree generated by a given scheme satisfies a given \msoufin sentence.
	\end{thm}

	We remark that the \msoufin logic extends the \wmsou logic,
	which was widely considered in the context of infinite words~\cite{wmso+u-words} and infinite trees~\cite{wmso+u-kaiser, wmso+u-trees, wmso+up}.
	The difference is that in \wmsou only quantification over finite sets is allowed.
	In consequence, \wmsou cannot express all properties of \mso~\cite{hummel-skrzypczak-topological}.
	On the other hand, in \msoufin we allow quantification over infinite sets like in standard \mso,
	and we only restrict the use of the \unbound quantifier to subformulae in which all free variables represent finite sets.

	Furthermore, we remark that some restriction for the \unbound quantifier is necessary.
	Indeed, the model-checking problem for the full \msou logic (where the \unbound quantifier can be used in an unrestricted way)
	is undecidable already over the infinite word without labels~\cite{mso+u-undecid},
	so even more over all fancy trees that can be generated by higher-order recursion schemes.

	While proving Theorem~\ref{thm:main}, we depend on several earlier results.
	First, we translate \msoufin formulae to an equivalent automata model using the notion of logical types (aka.\ composition method) following a long series of previous work
	(some selection:~\cite{FefermanVaught, Shelah, Lauchli, BlumensathColcombet, wmso+u-kaiser, wmso-model}).
	Second, we use the logical-reflection property of schemes~\cite{reflection}.
	It says that given a scheme $\Gg$ and an \mso sentence $\phi$ one can construct a scheme $\Gg_\phi$ generating the same tree as $\Gg$,
	where in every node it is additionally written whether $\phi$ is satisfied in the subtree starting in this node.
	Third, we use an analogous property for the simultaneous unboundedness problem, called \emph{SUP reflection}~\cite{types-diagonal-journal}:
	given a scheme $\Gg$ we can construct a scheme $\Gg_\mathit{SUP}$ generating the same tree as $\Gg$,
	where every node is additionally annotated by the solution of the simultaneous unboundedness problem in the subtree starting in this node.
	Finally, we use the fact that schemes can be composed with finite tree transducers transforming the generated trees;
	this follows directly from the equivalence between schemes and collapsible pushdown systems~\cite{collapsible}.

	Although our algorithm depends on a solution to the simultaneous unboundedness problem, it is not known whether the simultaneous unboundedness problem itself can be expressed in \msoufin.
	The difficulty is that every \unbound quantifier can entail unboundedness only of a single quantity,
	and it seems difficult to express simultaneous unboundedness of multiple quantities in \msoufin.

	This paper is an extended version of a conference paper~\cite{wmsou-schemes}, where the result is shown for the \wmsou logic.
	Besides the fact that we work here with a slightly stronger logic (namely, with \msoufin instead of \wmsou),
	our proofs follow basically the same ideas as proofs contained in the conference paper~\cite{wmsou-schemes}.
	We remark that the conference paper~\cite{wmsou-schemes} contained additionally a justification of the SUP-reflection property for schemes.
	This justification was already expanded in another paper~\cite{types-diagonal-journal},
	and for this reason we do not include it here.

	Our paper is structured as follows.
	In Section~\ref{sec:prelim} we introduce all necessary definitions.
	In Section~\ref{sec:automata} we show how to translate \msoufin sentences to automata.
	In Section~\ref{sec:main-thm} we prove the main theorem.
	Section~\ref{sec:conclusion} contains a few extensions of the main theorem.

\section{Preliminaries}\label{sec:prelim}

	The powerset of a set $X$ is denoted $\Pp(X)$, and the set of finite subsets of $X$ is denoted $\Ppfin(X)$.
	For a relation $r$, we write $r^*$ for the reflexive transitive closure of $r$.
	When $f$ is a function, by $f[x\mapsto y]$ we mean the function that maps $x$ to $y$ and every other $z\in\dom(f)$ to $f(z)$.

\subsection*{Infinitary lambda-calculus.}
	We consider infinitary, simply-typed lambda-calculus.
	In particular, each lambda-term has an associated sort (aka.\ simple type).
	The set of \emph{sorts} is constructed from a unique ground sort $\otyp$ using a binary operation $\arr$; namely $\otyp$ is a sort, and if $\alpha$ and $\beta$ are sorts, so is $\alpha\arr\beta$.
	By convention, $\arr$ associates to the right, that is, $\alpha\arr\beta\arr\gamma$ is understood as $\alpha\arr(\beta\arr\gamma)$.

	While defining lambda-terms we assume %an infinite set of letters $\Sigma$ (we use unranked letters; this subsumes the setting of ranked letters), and
	a set of variables $\Vars^\lambda=\{x^\alpha,y^\beta,z^\gamma,\dots\}$ containing infinitely many variables of every sort (sort of a variable is written in superscript).
	\emph{Infinitary lambda-terms} (or just \emph{lambda-terms}) are defined by coinduction, according to the following rules:
	\begin{itemize}
	\item	node constructor---if $\tK_1^\otyp,\dots,\tK_\nChld^\otyp$ are lambda-terms, and $a$ is an arbitrary object, called a \emph{letter},
		then $(\symb{a}[\tK_1^\otyp][\dots][\tK_\nChld^\otyp])^\otyp$ is a lambda-term,
	\item	variable---every variable $x^\alpha\in\Vars^\lambda$ is a lambda-term,
	\item	application---if $\tK^{\alpha\arr\beta}$ and $\tL^\alpha$ are lambda-terms, then $(\tK^{\alpha\arr\beta}\,\tL^\alpha)^\beta$ is a lambda-term, and
	\item	lambda-binder---if $\tK^\beta$ is a lambda-term and $x^\alpha$ is a variable, then $(\lambda x^\alpha.\tK^\beta)^{\alpha\arr\beta}$ is a lambda-term.
	\end{itemize}
	Sets of letters are called \emph{alphabets}.
	We use unranked letters; this subsumes the setting of ranked letters.
	We naturally identify lambda-terms differing only in names of bound variables.
	We often omit the sort annotations of lambda-terms, but we keep in mind that every lambda-term (and every variable) has a fixed sort.
	Free variables and subterms of a lambda-term, as well as beta-reductions, are defined as usual.
	A lambda-term $\tK$ is \emph{closed} if it has no free variables.
	We restrict ourselves to those lambda-terms for which the set of sorts of all subterms is finite.

\subsection*{Trees; B\"ohm trees.}

	A \emph{tree} is defined as a lambda-term that is built using only node constructors, that is, not using variables, applications, nor lambda-binders.
	For a tree $T=\symb{a}[T_1][\dots][T_\nChld]$, its set of nodes is defined as the smallest set such that
	\begin{itemize}
	\item	$\epsilon$ is a node of $T$, labeled by $a$, and
	\item	if $\unode$ is a node of $T_i$ for some $i\in\{1,\dots,\nChld\}$, labeled by $b$, then $i\unode$ is a node of $T$, also labeled by $b$.
	\end{itemize}
	A node $\vnode$ is the \emph{$i$-th child} of $\unode$ if $\vnode=\unode i$.
	We say that two trees $T,T'$ are of \emph{the same shape} if they have the same nodes.
	A tree $T$ is \emph{over alphabet $\Sigma$} if all labels of its nodes belong to $\Sigma$,
	and it has \emph{maximal arity} $\rMax\in\Nat$ if its every node has at most $\rMax$ children.
	When both these conditions are satisfied, we say that $T$ is a \emph{$(\Sigma,\rMax)$-tree}.
	For a tree $T$ and its node $u$, by $T\restr_u$ we denote the \emph{subtree} of $T$ starting at $u$, defined in the expected way.

	We consider B\"ohm trees only for closed lambda-terms of sort $\otyp$.
	For such a lambda-term $\tK$, its \emph{B\"ohm tree} is constructed by coinduction, as follows:
	if there is a sequence of beta-reductions from $\tK$ to a lambda-term of the form $\symb{a}[\tK_1][\dots][\tK_\nChld]$,
	and $T_1,\dots,T_\nChld$ are B\"ohm trees of $\tK_1,\dots,\tK_\nChld$, respectively, then $\symb{a}[T_1][\dots][T_\nChld]$ is a B\"ohm tree of $\tK$;
	if there is no such sequence of beta-reductions from $\tK$, then $\symb{\omega}$ is a B\"ohm tree of $\tK$ (where $\omega$ is a fixed letter).
	It is folklore that every closed lambda-term of sort $\otyp$ has exactly one B\"ohm tree (the order in which beta-reductions are performed does not matter); this tree is denoted by $\BT(\tK)$.

	A closed lambda-term $\tK$ of sort $\otyp$ is called \emph{fully convergent} if every node of $\BT(\tK)$ is explicitly created by a node constructor from $\tK$
	(e.g., $\symb{\omega}$ is fully convergent, while $K=(\lambda x^\otyp.x)\,K$ is not).
	More formally: we consider the lambda-term $\tK_{-\omega}$ obtained from $\tK$ by replacing $\omega$ with some other letter $\omega'$, and
	we say that $\tK$ is fully convergent if in $\BT(\tK_{-\omega})$ there are no $\omega$-labeled nodes.

\subsection*{Higher-order recursion schemes.}

	Our definition of schemes is less restrictive than usually, as we see them only as finite representations of infinite lambda-terms.
	Thus a \emph{higher-order recursion scheme} (or just a \emph{scheme}) is a triple $\Gg=(\Nn,\Rr,\NIni^\otyp)$, where
	\begin{itemize}
	\item
	$\Nn\subseteq\Vars^\lambda$ is a finite set of nonterminals,
	\item
	$\Rr$ is a function that maps every nonterminal $N\in\Nn$ to a finite lambda-term whose all free variables are contained in $\Nn$ and whose sort equals the sort of $N$,
	and \item $\NIni^\otyp\in\Nn$ is a starting nonterminal, being of sort $\otyp$.
	\end{itemize}
	We assume that elements of $\Nn$ are not used as bound variables, and that $\Rr(N)$ is not a nonterminal for any $N\in\Nn$.

	For a scheme $\Gg=(\Nn,\Rr,\NIni)$, and for a lambda-term $\tK$ whose free variables are contained in $\Nn$,
	we define the infinitary lambda-term \emph{represented by} $\tK$ with respect to $\Gg$, denoted $\Lambda_\Gg(\tK)$, by coinduction:
	to obtain $\Lambda_\Gg(\tK)$ we replace in $\tK$ every nonterminal $N\in\Nn$ with $\Lambda_\Gg(\Rr(N))$.
	Observe that $\Lambda_\Gg(\tK)$ is a closed lambda-term of the same sort as $\tK$.
	The infinitary lambda-term \emph{represented by} $\Gg$, denoted $\Lambda(\Gg)$, equals $\Lambda_\Gg(\NIni)$.

	By the \emph{tree generated by $\Gg$} we mean $\BT(\Lambda(\Gg))$.
	We write $\Sigma_\Gg$ for the finite alphabet containing $\omega$ and letters used in node constructors appearing in $\Gg$,
	and $\rMax(\Gg)$ for the maximal arity of node constructors appearing in $\Gg$.
	Clearly $\BT(\Lambda(\Gg))$ is a $(\Sigma_\Gg,\rMax(\Gg))$-tree.

	In our constructions it is convenient to consider only schemes representing fully-conver\-gent lambda-terms, which is possible due to the following standard result.

	\begin{factfactC}[\cite{haddad-fics,models-reflection}]\label{fact:convergent}
		For every scheme $\Gg$ we can construct a scheme $\Gg'$ generating the same tree as $\Gg$, and such that $\Lambda(\Gg')$ is fully convergent.
	\newqed\end{factfactC}

	\begin{exam}
		Consider the scheme $\Gg_1=(\{\Mnon^\otyp,\Nnon^{\otyp\arr\otyp}\},\Rr,\Mnon)$, where
		\begin{align*}
			\Rr(\Nnon)=\lambda x^\otyp.\symb{a}[x][\Nnon\,(\symb{b}[x])]\,,
			&&\mbox{and}&&
			\Rr(\Mnon)=\Nnon\,(\symb{c})\,.
		\end{align*}
		We obtain $\Lambda(\Gg_1)=\tK\,(\symb{c})$, where $\tK$ is the unique lambda-term for which it holds $\tK=\lambda x^\otyp.\symb{a}[x][\tK\,(\symb{b}[x])]$.
		The tree generated by $\Gg_1$ equals $\symb{a}[T_0][\symb{a}[T_1][\symb{a}[T_2][\dots]]]$, where $T_0=\symb{c}$ and $T_i=\symb{b}[T_{i-1}]$ for all $i\geq 1$.
	\end{exam}

	\begin{rem*}
		An usual definition of schemes is more restrictive than ours:
		it is required that $\Rr(N)$ is a of the form $\lambda x_1\lamdots\lambda x_\args.K$, where $K$ is of sort $\otyp$ and does not use any lambda-binders.
		We do not have this requirement, so possibly $\Rr(N)$ does not start with a full sequence of lambda-binders, and possibly some lambda-binders are nested deeper in the lambda-term.
		It is, though, not difficult to convert a scheme respecting only our definition to a scheme satisfying these additional requirements (at the cost of introducing more nonterminals).
		We can, for example, use a translation between schemes and $\lambda Y$-terms from Salvati and Walukiewicz~\cite{schemes-lY}:
		their translation from schemes to $\lambda Y$-terms works well with our definition of schemes,
		while the translation from $\lambda Y$-terms to schemes produces schemes respecting the more restrictive definition.

		Another difference is that in the definition of the B\"ohm tree we allow arbitrary beta-reductions, while it is sometimes assumed that only outermost beta-reductions are allowed.
		It is a folklore that these two definitions are equivalent.

		There is one more difference: we expand a scheme to an infinite lambda-term, and then we operate on this lambda-term,
		while often finite lambda-terms containing nonterminals are considered, and appearances of nonterminals are expanded only when needed.
		This is a purely syntactical difference.
	\end{rem*}

\subsection*{\msoufin.}

	For technical convenience, we use a syntax in which there are no first-order variables.
	It is easy to translate a formula from a more standard syntax to ours (at least when the maximal arity of considered trees is fixed).
	We assume two infinite sets of variables, $\Vv^\mathsf{fin}$ and $\Vv^\mathsf{inf}$, and we let $\Vv=\Vv^\mathsf{fin}\uplus\Vv^\mathsf{inf}$.
	Variables from $\Vv^\mathsf{fin}$ are used to quantify over finite sets, while variables from $\Vv^\mathsf{inf}$ over arbitrary (potentially infinite) sets.
	In the syntax of \msoufin we have the following constructions:
	\begin{align*}
		\phi::=a(\Xvar)\mid
			 \Xvar\child_i \Yvar\mid
			 \Xvar\subseteq \Yvar\mid
			\phi_1\land\phi_2\mid
			\neg\phi'\mid
			\exists \Zvar.\phi'\mid
			\efin \Fvar.\phi'\mid
			\unbound \Fvar.\phi'
	\end{align*}
	where $a$ is a letter, and $i\in\Nat_+$, and $\Xvar,\Yvar\in\Vv$, and $\Zvar\in\Vv^\mathsf{inf}$, and $\Fvar\in\Vv^\mathsf{fin}$.
	Free variables of a formula are defined as usual; in particular $\unbound \Fvar$ is a quantifier, hence it bounds the variable $\Fvar$.
	We impose the restriction that $\unbound\Fvar.\phi'$ can be used only when all free variables of $\phi'$ are from $\Vv^\mathsf{fin}$.

	The \mso logic is defined likewise, with the exception that the \unbound quantifier is disallowed.
	(The fact that a set of tree nodes is finite is expressible in MSO without using the $\efin$ quantifier,
	thus presence of this quantifier does not change the expressive power of MSO).

	We evaluate formulae of \msoufin in $\Sigma$-labeled trees.
	In order to evaluate a formula $\phi$ in a tree $T$, we also need a \emph{valuation}, that is, a partial function $\nu$ from $\Vv$ to sets of nodes of $T$,
	such that $\nu(\Fvar)$ is finite whenever $\Fvar\in\Vv^\mathsf{fin}\cap\dom(\nu)$.
	The function should be defined at least for all free variables of $\phi$.
	The semantics is defined as follows:
	\begin{itemize}
	\item	$a(\Xvar)$ holds when every node in $\nu(\Xvar)$ is labeled by $a$,
	\item	$\Xvar\child_i \Yvar$ holds when both $\nu(\Xvar)$ and $\nu(\Yvar)$ are singletons,
		and the unique node in $\nu(\Yvar)$ is the $i$-th child of the unique node in $\nu(\Xvar)$,
	\item	$\Xvar\subseteq \Yvar$ holds when $\nu(\Xvar)\subseteq\nu(\Yvar)$,
	\item	$\phi_1\land\phi_2$ holds when both $\phi_1$ and $\phi_2$ hold,
	\item	$\neg\phi'$ holds when $\phi'$ does not hold,
	\item	$\exists\Zvar.\phi'$ holds when $\phi'$ holds for an extension of $\nu$ that maps $\Zvar$ to some set of nodes of $T$,
	\item	$\efin \Fvar.\phi'$ holds when $\phi'$ holds for an extension of $\nu$ that maps $\Fvar$ to some finite set of nodes of $T$, and
	\item	$\unbound \Fvar.\phi'$ holds when for every $n\in\Nat$,
		$\phi'$ holds for an extension of $\nu$ that maps $\Fvar$ some finite set of nodes of $T$ of cardinality at least $n$.
	\end{itemize}
	We write $T,\nu\models\phi$ to denote that $\phi$ holds in $T$ with respect to the valuation $\nu$.

	In order to see that our definition of \msoufin is not too poor, let us write a few example formulae.
	\begin{itemize}
	\item	The fact that $\Xvar$ represents an empty set can be expressed as $\mathit{empty}(\Xvar)\equiv a(\Xvar)\land b(\Xvar)$ (where $a,b$ are any two different letters).
	\item	The fact that $\Xvar$ represents a set of size at least $2$
		can be expressed as $\mathit{big}(\Xvar)\equiv\exists \Yvar.(\Yvar\subseteq \Xvar\land \neg(\Xvar\subseteq \Yvar)\land\neg\mathit{empty}(\Yvar))$.
	\item	The fact that $\Xvar$ represents a singleton can be expressed as $\mathit{sing}(\Xvar)\equiv\neg\mathit{empty}(\Xvar)\land\neg\mathit{big}(\Xvar)$.
	\item	When we only consider trees of a fixed maximal arity $\rMax$,
		the fact that $\Xvar$ and $\Yvar$ represent singletons $\{x\},\{y\}$, respectively, such that $y$ is a child of $x$ can be expressed as
		\begin{align*}
			(\Xvar\child_1 \Yvar)\lor\dots\lor(\Xvar\child_{\rMax} \Yvar)\,,
		\end{align*}
		where $\phi_1\lor\phi_2$ stands for $\neg(\neg\phi_1\land\neg\phi_2)$.
	\item	Let $A=\{a_1,\dots,a_k\}$ be a finite set of letters.
		The fact every node in the set represented by $\Xvar$ has label in $A$ can be expressed as
		\begin{align*}
			\forall\Yvar.(\mathit{sing}(\Yvar)\land\Yvar\subseteq\Xvar)\arr(a_1(\Yvar)\lor\dots\lor a_k(\Yvar))\,,
		\end{align*}
		where $\forall\Yvar.\phi$ stands for $\neg\exists\Yvar.\neg\phi$, and $\phi_1\arr\phi_2$ stands for $\neg(\phi_1\land\neg\phi_2)$.
	\end{itemize}

	Like in most logics, in \msoufin we can relativize formulae, as described by Fact~\ref{fact:relativize}.

	\begin{factfact}\label{fact:relativize}
		Let $\rMax\in\Nat$.
		For every \msoufin sentence $\phi$ we can construct an \msoufin formula $\widehat\phi(\Xvar)$ with one free variable $\Xvar$
		such that for every tree $T$ of maximal arity $\rMax$ and every valuation $\nu$,
		it holds $T,\nu\models\widehat\phi$ if and only if $\phi$ holds in $T\restr_\unode$ for every $\unode\in\nu(\Xvar)$.
	\end{factfact}

	\begin{proof}[Proof sketch]
		Suppose first that we want to construct a formula $\phi'(\Xvar)$ that satisfies the fact only for valuations mapping $\Xvar$ to singleton sets $\{\unode\}$.
		To this end, we need to relativize quantification in $\phi$ to the subtree starting in $\unode$.
		This means that we replace subformulae of the form $\exists\Yvar.\psi$ (and likewise $\efin\Yvar.\psi$ and $\unbound\Yvar.\psi$) by $\exists\Yvar.\eta(\Xvar,\Yvar)\land\psi$,
		where $\eta(\Xvar,\Yvar)$ says that the set represented by $\Yvar$ contains only (not necessarily proper) descendants of the node represented by $\Xvar$.

		We conclude by taking $\widehat\phi(\Xvar)\equiv\forall\Xvar'.(\mathit{sing}(\Xvar')\land\Xvar'\subseteq\Xvar)\arr\phi'(\Xvar')$,
		saying that the formula $\phi'(\Xvar')$ holds whenever $\Xvar'$ represents a singleton subset of the set represented by $\Xvar$.
	\end{proof}

\section{Nested \unbound-prefix \mso automata}\label{sec:automata}

	In this section we give a definition of nested \unbound-prefix \mso automata, a formalism equivalent to the \msoufin logic.
	These are compositions of \unbound-prefix automata and \mso automata, defined below.

	A \emph{\unbound-prefix automaton} is a tuple $\Aa=(\Sigma,Q,\QImp,\Delta)$,
	where $\Sigma$ is a finite alphabet,
	$Q$ is a finite set of states, $\QImp\subseteq Q$ is a set of \emph{important} states,
	and $\Delta\subseteq Q\times\Sigma\times (Q\cup\{\top\})^*$ is a finite transition relation (we assume $\top\not\in Q$).
	A \emph{run} of $\Aa$ on a $\Sigma$-labeled tree $T$ is a mapping $\rho$ from the set of nodes of $T$ to $Q\cup\{\top\}$ such that
	\begin{itemize}
	\item	there are only finitely many nodes $\unode$ such that $\rho(\unode)\in Q$, and
	\item	for every node $\unode$ of $T$, with label $a$ and $\nChld$ children, it holds that either $\rho(\unode)=\top=\rho(\unode 1)=\dots=\rho(\unode\nChld)$ or $(\rho(\unode),a,\rho(\unode1),\dots,\rho(\unode\nChld))\in\Delta$.
	\end{itemize}
	We use \unbound-prefix automata as transducers, relabeling nodes of $T$:
	we define $\Aa(T)$ to be the tree of the same shape as $T$, and such that its every node $\unode$ originally labeled by $a_\unode$ becomes labeled by the pair $(a_\unode,f_\unode)$,
	where $f_\unode\colon Q\to\{0,1,2\}$ is the function that assigns to every state $q\in Q$
	\begin{itemize}
	\item	$2$, if for every $n\in\Nat$ there is a run $\rho_n$ of $\Aa$ on $T\restr_\unode$ that assigns $q$ to the root of $T\restr_\unode$, and such that for at least $n$ nodes $\wnode$ it holds that $\rho_n(\wnode)\in \QImp$;
	\item	$1$, if the above does not hold, but there is a run of $\Aa$ on $T\restr_\unode$ that assigns $q$ to the root of $T\restr_\unode$;
	\item	$0$, if none of the above holds.
	\end{itemize}

	\newcommand{\qfin}{q_\mathsf{fin}}
	\newcommand{\qfb}{q_{\exists\mathit{lf}}}%\mathsf{fb}}
	\begin{exam}\label{exa:3.1}
		Consider the \unbound-prefix automaton $\Aa_1=(\{a\},\{\qfb,\qfin\},\{\qfin\},\Delta)$,
		where $\Delta$ contains transitions
		\begin{align*}
			&(\qfin,a),(\qfin,a,\qfin),(\qfin,a,\qfin,\qfin),(\qfb,a),&&\mbox{and}\\
			&(\qfb,a,q),(\qfb,a,q,\top),(\qfb,a,\top,q)&&\mbox{for }q\in\{\qfb,\qfin\}\,.
		\end{align*}
		Suppose now that a $(\{a\},2)$-tree $T$ comes.
		When a state $\qfin$ is assigned to some node $\unode$ of $T$, then it has to be assigned as well to all descendants of $\unode$.
		Thus, there is a run of $\Aa_1$ on $T$ with state $\qfin$ in the root exactly when the tree is finite.
		This is because the definition of a run allows to assign states (other than $\top$) only to finitely many nodes of the tree.
		Going further, there is a run of $\Aa$ on $T$ with state $\qfb$ in the root exactly when there is a leaf in the tree.
		The run can assign $\qfb$ to all nodes on the branch leading to a selected leaf, and $\top$ to all other nodes.
		Alternatively, it can assign $\qfb$ to nodes on a branch leading to some node $\unode$,
		and then $\qfin$ to all descendants of $\unode$, assuming that the subtree starting in $\unode$ is finite.

		Let $B_i$ be the full binary tree of height $i$, for $i\in\Nat$.
		Let $T_1$ be the tree consisting of an infinite branch, with tree $B_i$ attached below the $i$-th child of the branch;
		that is, we take $T_i=\symb{a}[T_{i+1}][B_i]$ for $i\in\Nat_+$.
		By definition, $\Aa_1(T_1)$ has the same shape as $T_1$.
		Nodes inside all $B_i$ become relabeled to $(a,[\qfb\mapsto1,\qfin\mapsto1])$.
		This is because every subtree of every $B_i$ is finite and has a leaf.
		Moreover, the number of nodes of this subtree to which $\qfin$ is assigned is bounded by the size of the subtree
		(and hence we do not use the value $2$ in the new label).
		Nodes of the leftmost branch of $T_1$ are, in turn, relabeled to $(a,[\qfb\mapsto2,\qfin\mapsto0])$.
		The value $2$ in the $i$-th node of the branch means that for every $n\in\Nat$ there is a run $\rho_n$ on $T_i$ that assigns $\qfb$ to the root of $T_i$,
		and assigns $\qfin$ to at least $n$ nodes.
		Such a run $\rho_n$ assigns $\qfb$ on a branch entering some $B_j$ with at least $n$ nodes,
		and assigns $\qfin$ to all nodes of this~$B_j$.
	\end{exam}

	An \emph{\mso automaton} is a triple $\Aa=(\Sigma,Q,(\phi_q)_{q\in Q})$, where $\Sigma$ is a finite alphabet,
	$Q$ is a finite set of states,
	and $(\phi_q)_{q\in Q}$ is a bundle of MSO sentences indexed by elements of $Q$.
	An effect of running such an automaton $\Aa$ on a $\Sigma$-labeled tree $T$ is the tree $\Aa(T)$ that is of the same shape as $T$,
	and such that its every node $\unode$ originally labeled by $a_\unode$ becomes labeled by the pair $(a_\unode,f_\unode)$,
	where $f_\unode\colon Q\to\{0,1,2\}$ is the function that assigns to every index $q\in Q$
	\begin{itemize}
	\item	$1$ if $\phi_q$ is true in $T\restr_\unode$;
	\item	$0$ otherwise.\footnote{%}
			\mso automata never assign the value $2$;
			nevertheless, for uniformity between \unbound-prefix automata and \mso automata, we assume that the set of values is $\{0,1,2\}$.%{
		}
	\end{itemize}

	\begin{exam}
		Let $\Aa_2=(\{a\}\times\{0,1,2\}^{\{\qfb,\qfin\}},\{q_1,q_2\},(\phi_{q_1},\phi_{q_2}))$,
		where $\phi_{q_1}$ says that the second child of the root exists and is labeled by $(a,[\qfb\mapsto1,\qfin\mapsto1])$,
		and $\phi_{q_2}$ says that there exists an infinite branch with all nodes labeled by $(a,[\qfb\mapsto2,\qfin\mapsto0])$.
		Let us analyze $\Aa_2(\Aa_1(T_1))$, for the tree $T_1$ from Example~\ref{exa:3.1}.
		It has the same shape as $\Aa_1(T_1)$, and as $T_1$.
		All leaves become labeled by $((a,[\qfb\mapsto1,\qfin\mapsto1]),[q_1\mapsto0,q_2\mapsto0])$,
		other nodes of $B_i$ become labeled by $((a,[\qfb\mapsto1,\qfin\mapsto1]),[q_1\mapsto1,q_2\mapsto0])$,
		and nodes on the leftmost branch of the tree become labeled by $((a,[\qfb\mapsto2,\qfin\mapsto0]),[q_1\mapsto1,q_2\mapsto1])$.
	\end{exam}

	By the \emph{input alphabet} of $\Aa$, where $\Aa$ is either a \unbound-prefix automaton $(\Sigma,Q,\QImp,\Delta)$ or an \mso automaton $(\Sigma,Q,(\phi_q)_{q\in Q})$,
	we mean the set $\Sigmain(\Aa)=\Sigma$.
	The \emph{output alphabet} of $\Aa$ is $\Sigmaout(\Aa)=\Sigmain(\Aa)\times\{0,1,2\}^Q$.

	A \emph{nested \unbound-prefix \mso automaton} is a sequence $\Aa=\Aa_1\vartriangleright\dots\vartriangleright\Aa_k$ (with $k\geq 1$),
	where every $\Aa_i$ is either a \unbound-prefix automaton or an \mso automaton, and where $\Sigmain(\Aa_{i+1})=\Sigmaout(\Aa_i)$ for $i\in\{1,\dots,k-1\}$.
	We define $\Aa(T)$ to be $\Aa_k(\dots(\Aa_1(T))\dots)$.
	The input and output alphabets of $\Aa$, denoted $\Sigmain(\Aa)$ and $\Sigmaout(\Aa)$, equal $\Sigmain(\Aa_1)$ and $\Sigmaout(\Aa_k)$, respectively.
	The key property is that these automata can check properties expressed in \msoufin,
	as we state in Lemma~\ref{lem:logic-to-automata}, and we prove in the remainder of this section.

	\begin{lem}\label{lem:logic-to-automata}
		Let $\Sigma$ be a finite alphabet, and let $\rMax\in\Nat$.
		For every \msoufin sentence $\phi$ we can construct a nested \unbound-prefix \mso automaton $\Aa_\phi$ with $\Sigmain(\Aa_\phi)=\Sigma$,
		and a subset $\SigmaF\subseteq\Sigmaout(\Aa_\phi)$ such that for every $(\Sigma,\rMax)$-tree $T$,
		the root of $\Aa_\phi(T)$ is labeled by a letter in $\SigmaF$ if and only if $\phi$ holds in $T$.
	\end{lem}

	Recall that our aim is to evaluate $\phi$ in a tree $T$ generated by a recursion scheme $\Gg$, so the restriction to $(\Sigma,\rMax)$-trees is not harmful:
	as $(\Sigma,\rMax)$ we are going to take $(\Sigma_\Gg,\rMax(\Gg))$.

	It is not difficult to see that in \msoufin we can express properties checked by nested \unbound-prefix MSO automata.
	This means that the two formalisms are actually equivalent.
	%Actually, the two formalisms are equivalent: in \msoufin we can express properties described by our automata.
	Although we do not need this second direction in order to prove Theorem~\ref{thm:main}, we state it in Lemma~\ref{lem:automata-to-logic} for cognitive purposes.

	\begin{lem}\label{lem:automata-to-logic}
		Let $\rMax\in\Nat$.
		For every nested \unbound-prefix \mso automaton $\Aa$, and every letter $\eta\in\Sigmaout(\Aa)$ we can construct an \msoufin sentence $\phi_{\Aa,\eta}$
		such that for every $(\Sigmain(\Aa),\rMax)$-tree $T$, the root of $\Aa(T)$ is labeled by $\eta$ if and only if $\phi_{\Aa,\eta}$ holds in $T$.
	\end{lem}

	\begin{proof}[Proof sketch]
		When $\Aa$ is a single \mso automaton, $\Aa=(\Sigma,Q,(\psi_q)_{q\in Q})$, it is straightforward to construct $\phi_{\Aa,\eta}$ in question.
		Namely, when $\eta=(a,f)$, as $\phi_{\Aa,\eta}$ we take
		\begin{align*}
			\xi_a\land\bigwedge_{q:f(q)=1}\psi_q\land\bigwedge_{q:f(q)=0}\neg\psi_q\,,
		\end{align*}
		where $\xi_a$ says that the root is labeled by $a$.

		It is also not difficult to deal with a single \unbound-prefix automaton.
		Indeed, it is standard to express in \mso that a run of an automaton exists.
		The fact that there exist runs with arbitrarily many important states is expressed using the \unbound quantifier.

		It remains to simulate composition of automata.
		Suppose that $\Aa=\Aa_1\vartriangleright\Aa_2$ (where $\Aa_1,\Aa_2$ may be nested again),
		and that we already have sentences $\psi_{\Aa_1,a}$ corresponding to $\Aa_1$ for $a\in\Sigmaout(\Aa_1)=\Sigmain(\Aa_2)$,
		and $\psi_{\Aa_2,\eta}$ corresponding to $\Aa_2$.
		Out of every sentence $\psi_{\Aa_1,a}$ we construct a formula $\widehat\psi_{\Aa_1,a}(\Zvar)$ saying that $\psi_{\Aa_1,a}$
		holds in all subtrees starting in elements of the set represented by $\Zvar$ (cf.~Fact~\ref{fact:relativize}).
		The formula $\psi_{\Aa_2,\eta}$ is evaluated in $\Aa_1(T)$, while the formula $\phi_{\Aa,\eta}$ that we are going to construct is evaluated in $T$.
		Thus, whenever $\psi_{\Aa_2,\eta}$ uses $a(\Zvar)$ for some letter $a\in\Sigmain(\Aa_2)$ and some variable $\Zvar$,
		in $\phi_{\Aa,\eta}$ we replace it by $\widehat\psi_{\Aa_1,a}(\Zvar)$.
	\end{proof}

	We now come to the proof of Lemma~\ref{lem:logic-to-automata}.
	We notice that due to the nested structure, our automata are quite close to the logic.
	It is clear that \mso automata can simulate all of \mso.
	On the other hand, $\unbound$-prefix automata check whether something is unbounded, which corresponds to \unbound quantifiers.
	As states of the $\unbound$-prefix automata we take \emph{phenotypes} (aka.\ logical types), which are defined next.

	Let $\phi$ be a formula of \msoufin, let $T$ be a tree, and let $\nu$ be a valuation (defined at least for all free variables of $\phi$).
	We define the \emph{$\phi$-phenotype} of $T$ under valuation $\nu$, denoted $\pht{\phi}{T}{\nu}$, by induction on the size of $\phi$ as follows:
	\begin{itemize}
	\item	if $\phi$ is of the form $a(\Xvar)$ (for some letter $a$) or $\Xvar\subseteq\Yvar$ then $\pht{\phi}{T}{\nu}$ is the logical value of $\phi$ in $T,\nu$, that is,
		$\true$ if $T,\nu\models\phi$ and $\false$ otherwise,
	\item	if $\phi$ is of the form $\Xvar\child_i \Yvar$, then $\pht{\phi}{T}{\nu}$ equals
		\begin{itemize}
		\item	$\true$ if $T,\nu\models\phi$,
		\item	$\mathsf{empty}$ if $\nu(\Xvar)=\nu(\Yvar)=\emptyset$,
		\item	$\mathsf{root}$ if $\nu(\Xvar)=\emptyset$ and $\nu(\Yvar)=\{\epsilon\}$, and
		\item	$\false$ otherwise,
		\end{itemize}
	\item	if $\phi\equiv(\psi_1\land\psi_2)$, then $\pht{\phi}{T}{\nu}=(\pht{\psi_1}{T}{\nu},\pht{\psi_2}{T}{\nu})$,
	\item	if $\phi\equiv(\neg\psi)$, then $\pht{\phi}{T}{\nu}=\pht{\psi}{T}{\nu}$,
	\item	if $\phi\equiv\exists\Xvar.\psi$, then
		\begin{align*}
			\pht{\phi}{T}{\nu}=\{\sigma\mid\exists X.\pht{\psi}{T}{\nu[\Xvar\mapsto X]}=\sigma\}\,,
		\end{align*}
	\item	if $\phi\equiv\efin \Xvar.\psi$, then
		\begin{align*}
			\pht{\phi}{T}{\nu}=\{\sigma\mid\exists X.\pht{\psi}{T}{\nu[\Xvar\mapsto X]}=\sigma\land |X|<\infty\}\,,\tag*{\mbox{and}}
		\end{align*}
	\item	if $\phi\equiv\unbound \Xvar.\psi$, then
		\begin{align*}
			\pht{\phi}{T}{\nu}=(&\{\sigma\mid\exists X.\pht{\psi}{T}{\nu[\Xvar\mapsto X]}=\sigma\land |X|<\infty\},\\
				&\{\sigma\mid\forall n\in\Nat.\exists X.\pht{\psi}{T}{\nu[\Xvar\mapsto X]}=\sigma\land n\leq|X|<\infty\})\,,
		\end{align*}
	\end{itemize}
	where $X$ ranges over sets of nodes of $T$.

	For each $\phi$, let $\Pht\phi$ denote the set of all potential $\phi$-phenotypes.
	Namely, $\Pht\phi=\{\true,\false\}$ in the first case,
	$\Pht\phi=\{\true, \mathsf{empty}, \mathsf{root}, \false\}$ in the second case,
	$\Pht\phi=\Pht{\psi_1}\times\Pht{\psi_2}$ in the third case,
	$\Pht\phi=\Pht{\psi}$ in the fourth case,
	$\Pht\phi=\Pp(\Pht\psi)$ in the fifth and sixth case, and
	$\Pht\phi=(\Pp(\Pht\psi))^2$ in the last case.

	We immediately see two facts.
	First, $\Pht\phi$ is finite for every $\phi$.
	Second, the fact whether $\phi$ holds in $T,\nu$ is determined by $\pht{\phi}{T}{\nu}$.
	This means that there is a function $\mathit{tv}_\phi\colon\Pht\phi\to\{\true,\false\}$ such that $\mathit{tv}_\phi(\pht{\phi}{T}{\nu})=\true$ if and only if $T,\nu\models
	\phi$.\label{page:tv}

	Next, we observe that phenotypes behave in a compositional way, as formalized below.
	Here for a valuation $\nu$ and a node $\unode$, by $\nu\restr_\unode$ we mean the valuation that restricts $\nu$ to the subtree starting at $\unode$,
	that is, maps every variable $\Xvar\in\dom(\nu)$ to $\{\wnode\mid\unode\wnode\in\nu(\Xvar)\}$.

	\begin{lemC}[\cite{wmso+u-kaiser,wmso-model}]\label{lem:compositionality}
		For every letter $a$, every $\nChld\in\Nat$, and every \msoufin formula $\phi$,
		one can compute a function $\Comp_{a,\nChld,\phi}\colon\Ppfin(\Vv)\times(\Pht\phi)^\nChld\to\Pht\phi$ such that
		for every tree $T$ whose root has label $a$ and $\nChld$ children, and for every valuation $\nu$,
		\begin{align*}
			\pht{\phi}{T}{\nu}=\Comp_{a,\nChld,\phi}(\{\Xvar\in\dom(\nu)\mid\epsilon\in\nu(\Xvar)\},\pht{\phi}{T\restr_1}{\nu\restr_1},\dots,\pht{\phi}{T\restr_\nChld}{\nu\restr_\nChld})\,.
		\end{align*}
	\end{lemC}

	\begin{proof}
		We proceed by induction on the size of $\phi$.

		When $\phi$ is of the form $b(\Xvar)$ or $\Xvar\subseteq \Yvar$, then we see that $\phi$ holds in $T,\nu$ if and only if it holds in every subtree $T\restr_i,\nu\restr_i$ and in the root of $T$.
		Thus, for $\phi\equiv b(\Xvar)$ as $\Comp_{a,\nChld,\phi}(R,\tau_1,\dots,\tau_\nChld)$ we take $\true$ when $\tau_i=\true$ for all $i\in\{1,\dots,\nChld\}$ and either $a=b$ or $\Xvar\not\in R$.
		For $\phi\equiv(\Xvar\subseteq \Yvar)$ the last part of the condition is replaced by ``if $\Xvar\in R$ then $\Yvar\in R$''.

		Next, suppose that $\phi\equiv(\Xvar\child_k \Yvar)$.
		Then as $\Comp_{a,\nChld,\phi}(R,\tau_1,\dots,\tau_\nChld)$ we take
		\begin{itemize}
		\item	$\true$ if $\tau_j=\true$ for some $j\in\{1,\dots,\nChld\}$, and $\tau_i=\mathsf{empty}$ for all $i\in\{1,\dots,\nChld\}\setminus\{j\}$, and $\Xvar\not\in R$, and $\Yvar\not\in R$,
		\item	$\true$ also if $\tau_k=\mathsf{root}$, and $\tau_i=\mathsf{empty}$ for all $i\in\{1,\dots,\nChld\}\setminus\{k\}$, and $\Xvar\in R$, and $\Yvar\not\in R$,
		\item	$\mathsf{empty}$ if $\tau_i=\mathsf{empty}$ for all $i\in\{1,\dots,\nChld\}$, and $\Xvar\not\in R$, and $\Yvar\not\in R$,
		\item	$\mathsf{root}$ if $\tau_i=\mathsf{empty}$ for all $i\in\{1,\dots,\nChld\}$, and $\Xvar\not\in R$, and $\Yvar\in R$, and
		\item	$\false$ otherwise.
		\end{itemize}
		By comparing this definition with the definition of the phenotype we immediately see that the thesis is satisfied.

		When $\phi\equiv(\neg\psi)$, we simply take $\Comp_{a,\nChld,\phi}=\Comp_{a,\nChld,\psi}$, and when $\phi\equiv(\psi_1\land\psi_2)$, as $\Comp_{a,\nChld,\phi}(R,(\tau_1^1,\tau_1^2),\dots,(\tau^1_\nChld,\tau^2_\nChld))$ we take the pair of $\Comp_{a,\nChld,\psi_i}(R,\tau_1^i,\dots,\tau^i_\nChld)$ for $i\in\{1,2\}$.

		Suppose now that $\phi\equiv\exists \Xvar.\psi$ or $\phi\equiv\efin \Xvar.\psi$.
		As $\Comp_{a,\nChld,\phi}(R,\tau_1,\dots,\tau_\nChld)$ we take
		\begin{align*}
			&\{\Comp_{a,\nChld,\psi}(R\cup\{\Xvar\},\sigma_1,\dots,\sigma_\nChld),\Comp_{a,\nChld,\psi}(R\setminus\{\Xvar\},\sigma_1,\dots,\sigma_\nChld)\\
				&\hspace{20em}\mid(\sigma_1,\dots,\sigma_\nChld)\in \tau_1\times\dots\times\tau_\nChld\}\,.
		\end{align*}
		The two possibilities, $R\cup\{\Xvar\}$ and $R\setminus\{\Xvar\}$, correspond to the fact that when quantifying over $\Xvar$,
		the root of $T$ may be either taken to the set represented by $\Xvar$ or not.
		Notice that the cases of $\exists\Xvar$ and $\efin\Xvar$ are handled in the same way:
		for a local behavior near the root it does not matter whether we quantify over all sets or only over finite sets.

		Finally, suppose that $\phi\equiv\unbound \Xvar.\psi$.
		The arguments of $\Comp_{a,\nChld,\phi}$ are pairs $(\tau_1,\rho_1),\dots,\allowbreak(\tau_\nChld,\rho_\nChld)$.
		Let $A$ be the set of tuples $(\sigma_1,\dots,\sigma_\nChld)\in\tau_1\times\dots\times\tau_\nChld$,
		and let $B$ be the set of tuples $(\sigma_1,\dots,\sigma_\nChld)$ such that $\sigma_j\in\rho_j$ for some $j\in\{1,\dots,\nChld\}$ and $\sigma_i\in\tau_i$ for all $i\in\{1,\dots,\nChld\}\setminus\{j\}$.
		As $\Comp_{a,\nChld,\phi}(R,(\tau_1,\rho_1),\dots,(\tau_\nChld,\rho_\nChld))$ we take
		\begin{align*}
			(&\{\Comp_{a,\nChld,\psi}(R\cup\{\Xvar\},\sigma_1,\dots,\sigma_\nChld),\Comp_{a,\nChld,\psi}(R\setminus\{\Xvar\},\sigma_1,\dots,\sigma_\nChld)\mid(\sigma_1,\dots,\sigma_\nChld)\in A\},\\
			 &\{\Comp_{a,\nChld,\psi}(R\cup\{\Xvar\},\sigma_1,\dots,\sigma_\nChld),\Comp_{a,\nChld,\psi}(R\setminus\{\Xvar\},\sigma_1,\dots,\sigma_\nChld)\mid(\sigma_1,\dots,\sigma_\nChld)\in B\})\,.
		\end{align*}
		The first coordinate is defined as for the existential quantifiers.
		The second coordinate is computed correctly due to the pigeonhole principle: if for every $n$ we have a set $X_n$ of cardinality at least $n$ (satisfying some property),
		then we can choose an infinite subsequence of these sets such that either the root belongs to all of them or to none of them,
		and one can choose some $j\in\{1,\dots,\nChld\}$ such that the sets contain unboundedly many descendants of~$j$.
	\end{proof}

	In order to prove Lemma~\ref{lem:logic-to-automata} by induction on the structure of the sentence $\phi$, we need to generalize it a bit; this is done in Lemma~\ref{lem:logic-to-automata-aux}.
	In particular, we need to use phenotypes, instead of the truth value of the sentence (because phenotypes are compositional, unlike truth values).
	We also need to allow formulae with free variables, not just sentences, as well as arbitrary valuations.
	A special role is played by the valuation $\VALempty$ that maps every variable to the empty set; for this valuation we have a stronger version of the lemma.

	\begin{lem}\label{lem:logic-to-automata-aux}
		Let $\Sigma$ be a finite alphabet, and let $\rMax\in\Nat$.
		Then for every \msoufin formula $\phi$ we can construct
		\begin{enumerate}
		\item	a nested \unbound-prefix \mso automaton $\Aa_\phi$ with $\Sigmain(\Aa_\phi)=\Sigma$, and \mso formulae $\xi_{\phi,\tau}$ for all $\tau\in\Pht{\phi}$,
			such that for every $(\Sigma,\rMax)$-tree $T$, every valuation $\nu$ in $T$, and every $\tau\in\Pht{\phi}$
			it holds $\Aa_\phi(T),\nu\models\xi_{\phi,\tau}$ if and only if $\pht{\phi}{T}{\nu}=\tau$, and
		\item	a nested \unbound-prefix \mso automaton $\Bb_\phi$ with $\Sigmain(\Bb_\phi)=\Sigma$, and a function $f_\phi\colon\allowbreak\Sigmaout(\Bb_\phi)\to\Pht{\phi}$,
			such that for every $(\Sigma,\rMax)$-tree $T$ the root of $\Bb_\phi(T)$ is labeled by a letter $\eta$ such that $f_\phi(\eta)=\pht{\phi}{T}{\VALempty}$.
		\end{enumerate}
	\end{lem}

	\begin{proof}
		Induction on the size of $\phi$.
		We start by observing how Item (2) follows from Item (1).
		Item (1) gives us an automaton $\Aa_\phi$ and \mso formulae $\xi_{\phi,\tau}$ for all $\tau\in\Pht{\phi}$.
		We change these formulae into sequences $\xi_{\phi,\tau}'$, assuming that all their free variables are valuated to the empty set.
		More precisely, for every free variable $\Xvar$ we change subformulae of $\xi_{\phi,\tau}$ of the form $a(\Xvar)$ and $\Xvar\subseteq\Yvar$ into $\true$,
		subformulae of the form $\Xvar\child_i \Yvar$ and $\Yvar\child_i \Xvar$ into $\false$,
		and subformulae of the form $\Yvar\subseteq\Xvar$, where $\Yvar$ is a bound variable, into formulae checking that the set represented by $\Yvar$ is empty.
		Then, we take $\Bb_\phi=\Aa_\phi\vartriangleright\Cc$ for $\Cc=(\Sigmaout(\Aa_\phi),\Pht{\phi},(\xi_{\phi,\tau}')_{\tau\in\Pht{\phi}})$.
		If $\eta=(a,h)$ for a function $h$ mapping exactly one phenotype $\tau$ to $1$, we define $f_\phi(\eta)$ to be this phenotype $\tau$,
		and for $\eta=(a,h)$ with $|h^{-1}(1)|\neq1$ we define $f_\phi(\eta)$ arbitrarily.

		Consider now a $(\Sigma,\rMax)$-tree $T$.
		By Item (1), for $\tau=\pht{\phi}{T}{\VALempty}$ we have $\Aa_\phi(T),\VALempty\models\xi_{\phi,\tau}$ (equivalently, $\Aa_\phi(T)\models\xi_{\phi,\tau}'$)
		and for $\tau\in\Pht{\phi}\setminus\{\pht{\phi}{T}{\VALempty}\}$ we have $\Aa_\phi(T),\VALempty\not\models\xi_{\phi,\tau}$ (equivalently, $\Aa_\phi(T)\not\models\xi_{\phi,\tau}'$).
		It follows that the root of $\Bb_\phi(T)$ is labeled by $\eta=(a,h)$ where $a$ is the label of the root in $\Aa_\phi(T)$, and $h(\pht{\phi}{T}{\VALempty})=1$,
		and $h(\tau)=0$ for $\tau\in\Pht{\phi}\setminus\{\pht{\phi}{T}{\VALempty}\}$.
		Then $f_\phi(\eta)=\pht{\phi}{T}{\VALempty}$, as required.

		We now come to the proof of Item (1), where we proceed by case distinction.
		When $\phi$ is an atomic formula, that is, equals $a(\Xvar)$, $\Xvar\subseteq\Yvar$, or $\Xvar\child_i\Yvar$, then the automaton $\Aa_\phi$ is not needed:
		as $\Aa_\phi$ we can take the MSO automaton with empty set of states (and input alphabet $\Sigma$).
		For such an automaton we have that $\Aa_\phi(T)=T$ for every $(\Sigma,\rMax)$-tree $T$.
		As $\xi_{\phi,\true}$ we take $\phi$.
		When $\phi$ equals $a(\Xvar)$ or $\Xvar\subseteq\Yvar$, the only phenotypes are $\true$ and $\false$, and thus we take $\xi_{\phi,\false}\equiv\neg\phi$.
		In the case of $\phi\equiv\Xvar\child_i\Yvar$, the situation when the formula is false is divided into three phenotypes: $\mathsf{empty}$, $\mathsf{root}$, and $\false$.
		Nevertheless, it is easy to express in \mso that we have a particular phenotype, following the definition of $\pht{\phi}{T}{\nu}$.

		Suppose now that $\phi$ is of the form $\psi_1\land\psi_2$.
		From the induction assumption, Item (1), we have two automata, $\Aa_{\psi_1}$ and $\Aa_{\psi_2}$,
		as well as formulae $\xi_{\psi_1,\tau_1}$ for all $\tau_1\in\Pht{\psi_1}$ and $\xi_{\psi_2,\tau_2}$ for all $\tau_2\in\Pht{\psi_2}$.
		We combine the two automata into a single automaton $\Aa_\phi$ with $\Sigmain(\Aa_\phi)=\Sigma$.
		More precisely, we take $\Aa_\phi=\Aa_{\psi_1}\vartriangleright\Aa_{\psi_2}'$, where $\Aa_{\psi_2}'$ works exactly like $\Aa_{\psi_2}$,
		but instead of reading a tree $T$ over alphabet $\Sigma$, it reads the tree $\Aa_{\psi_1}(T)$ and ignores the part of its labels added by $\Aa_{\psi_1}$.
		We also amend $\xi_{\psi_1,\tau_1}$ and $\xi_{\psi_2,\tau_2}$ so that they can read the output of $\Aa_\phi$:
		formulae $\xi_{\psi_1,\tau_1}'$ work like $\xi_{\psi_1,\tau_1}$ but ignore the parts of labels added by $\Aa_{\psi_2}'$,
		and formulae $\xi_{\psi_2,\tau_2}'$ work like $\xi_{\psi_2,\tau_2}$ but ignore the parts of labels added by $\Aa_{\psi_1}$.
		Having the above, for every $\tau=(\tau_1,\tau_2)\in\Pht{\phi}$ we take $\xi_{\phi,\tau}\equiv\xi_{\psi_1,\tau_1}'\land\xi_{\psi_2,\tau_2}'$.
		Because a tree has $\phi$-phenotype $\tau$ when it has $\psi_1$-phenotype $\tau_1$ and simultaneously $\psi_2$-phenotype $\tau_2$,
		it should be clear that the thesis of Item (1) becomes satisfied.

		When $\phi$ is of the form $\neg\psi$, or $\exists\Xvar.\psi$, or $\efin\Xvar.\psi$, as $\Aa_\phi$ we take $\Aa_\psi$ existing by the induction assumption, Item (1).
		The induction assumption gives us also formulae $\xi_{\psi,\tau}$ for all $\tau\in\Pht\psi$.
		If $\phi\equiv\neg\psi$, we take $\xi_{\phi,\tau}\equiv\xi_{\psi,\tau}$.
		If $\phi\equiv\exists\Xvar.\psi$, we take
		\begin{align*}
			\xi_{\phi,\tau}\equiv\bigwedge_{\sigma\in\tau}(\exists\Xvar.\xi_{\psi,\sigma})\land\bigwedge_{\sigma\in\Pht{\psi}\setminus\tau}(\neg\exists\Xvar.\xi_{\psi,\sigma})\,.
		\end{align*}
		If $\phi\equiv\efin\Xvar.\psi$, we take the same formula, but with $\efin$ quantifiers instead of $\exists$.

		Finally, suppose that $\phi\equiv\unbound \Xvar.\psi$.
		We cannot proceed like in the previous cases, because the \unbound quantifier cannot be expressed in \mso;
		we rather need to append a new \unbound-prefix automaton at the end of the constructed automaton.
		In this case we first prove Item (2), and then we deduce Item (1) out of Item (2).
		By Item (2) of the induction assumption we have an automaton $\Bb_\psi$ and a function $f_\psi\colon\Sigmaout(\Bb_\psi)\to\Pht{\psi}$ such that
		for every node $\unode$ of $T$, the root of $\Bb_\psi(T\restr_\unode)$ is labeled by a letter $\eta_\unode$ such that $f_\psi(\eta_\unode)=\pht{\psi}{T\restr_\unode}{\VALempty}$.
		Moreover, there is a function $g\colon\Sigmaout(\Bb_\psi)\to\Sigma$ such that $g(\eta_\unode)$ is the original label of $\unode$ in $T$
		(such a function exists, because the labels from $T$ remain as a part of the labels in $\Bb_\psi(T)$).
		Recall that $\Bb_\psi(T)$ has the same shape as $T$, and actually $(\Bb_\psi(T))\restr_\unode=\Bb_\psi(T\restr_\unode)$ for every node $\unode$.
		We construct a new layer $\Cc$, which calculates $\phi$-phenotypes basing on $\psi$-phenotypes, and we take $\Bb_\phi=\Bb_\psi\vartriangleright\Cc$.
		As the state set of $\Cc$ we take $Q=\{0,1\}\times\Pht{\psi}$; states from $\{1\}\times\Pht{\psi}$ are considered as important.
		Transitions are determined by the $\Comp$ predicate from Lemma~\ref{lem:compositionality}.
		More precisely, for every $\nChld\leq\rMax$, every $\eta\in\Sigmaout(\Bb_\psi)$, and all $((i_1,\sigma_1),\dots,(i_\nChld,\sigma_\nChld))\in Q^\nChld$ we have transitions
		\begin{align*}
			&((0,\Comp_{g(\eta),\nChld,\psi}(\emptyset,\sigma_1,\dots,\sigma_\nChld)),\eta,(i_1,\sigma_1),\dots,(i_\nChld,\sigma_\nChld))\,,&\mbox{and}\\
			&((1,\Comp_{g(\eta),\nChld,\psi}(\{\Xvar\},\sigma_1,\dots,\sigma_\nChld)),\eta,(i_1,\sigma_1),\dots,(i_\nChld,\sigma_\nChld))\,.
		\end{align*}
		Moreover, we have transitions that read the $\psi$-phenotype from the label:
		\begin{align*}
			&((0,f_\psi(\eta)),\eta,\underbrace{\top,\dots,\top}_\nChld)&\mbox{for }\nChld\leq\rMax.
		\end{align*}
		We notice that there is a direct correspondence between runs of $\Cc$ and choices of a set of nodes $X$ to which the variable $\Xvar$ is mapped.
		The first coordinate of the state is set to $1$ in nodes chosen to belong to the set $X$.
		The second coordinate contains the $\psi$-phenotype under the valuation mapping $\Xvar$ to $X$ and every other variable to the empty set.
		In some nodes below the chosen set $X$ we use transitions of the second kind, reading the $\psi$-phenotype from the label;
		it does not matter in which nodes this is done, as everywhere a correct $\psi$-phenotype is written.
		The fact that we quantify only over finite sets $X$ corresponds to the fact that the run of $\Cc$ can assign non-$\top$ states only to a finite prefix of the tree.
		Moreover, the cardinality of $X$ is reflected by the number of important states assigned by a run.
		It follows that for every $\sigma\in\Pht{\psi}$,
		\begin{itemize}
		\item	there exists a finite set $X$ of nodes of $T$ such that $\pht{\psi}{T}{\VALempty[\Xvar\mapsto X]}=\sigma$ if and only if for some $i\in\{0,1\}$ there is a run of $\Cc$ on $\Bb_\psi(T)$ that assigns $(i,\sigma)$ to the root, and
		\item	for every $n\in\Nat$ there exists a finite set $X_n$ of nodes of $T$ such that $\pht{\psi}{T}{\VALempty[\Xvar\mapsto X_n]}=\sigma$ and $|X_n|\geq n$
			if and only if for some $i\in\{0,1\}$ and for every $n\in\Nat$ there is a run $\rho_n$ of $\Cc$ on $\Bb_\psi(T)$ that assigns $(i,\sigma)$ to the root,
			and such that $\rho_n$ assigns an important state to at least $n$ nodes.
		\end{itemize}
		Thus, looking at the root's label in $\Bb_\phi(T)$ we can determine $\pht{\phi}{T}{\VALempty}$.
		This finishes the proof of Item (2) in the case of the \unbound quantifier.

		Next, still supposing that $\phi\equiv\unbound \Xvar.\psi$, we prove Item (1) using Item (2), which is already proved.
		It is essential that, by the definition of the \msoufin logic, all free variables of $\phi$ come from $\Vv^\mathsf{fin}$, that is, represent finite sets.
		This means that only nodes from a finite prefix of a considered tree can belong to $\nu(\Yvar)$ for $\Yvar$ free in $\phi$ (since clearly the number of free variables is finite).
		Outside of this finite prefix we can read the $\phi$-phenotype from the output of $\Bb_\phi$ (because the valuation is empty there),
		and in the finite prefix we can compute them using the $\Comp$ function.

		More precisely, as $\Aa_\phi$ we take $\Bb_\phi$, coming from Item (2).
		Item (2) gives us a function $f_\phi\colon\Sigmaout(\Aa_\phi)\to\Pht\phi$ reading the $\phi$-phenotype of a $(\Sigma,\rMax)$-tree $T$ out of the root label of $\Aa_\phi(T)$;
		we also have a function $g\colon\Sigmaout(\Aa_\phi)\to\Sigma$ that extracts original labels out of labels in $\Aa_\phi(T)$.
		For every $\tau\in\Pht\phi$ we define the formula $\xi_{\psi,\tau}$ as follows.
		It starts with a sequence of $|\Pht{\phi}|$ existential quantifiers,
		quantifying over variables $\Xvar_\rho$ for all $\rho\in\Pht{\phi}$.
		The intention is that, in a tree $T$, every $\Xvar_\rho$ represents the set of nodes $\unode$ such that $\pht{\phi}{T\restr_\unode}{\nu\restr_\unode}=\rho$.
		Inside the quantification we say that
		\begin{itemize}
		\item	the sets represented by these variables are disjoint, and every node belongs to some of them,
		\item	the root belongs to $\Xvar_\tau$,
		\item	if a node with label $\eta\in\Sigmaout(\Aa_\phi)$ belongs to $\Xvar_\rho$, and its children belong to $\Xvar_{\rho_1},\dots,\Xvar_{\rho_r}$, respectively
			(where $r\leq\rMax$), and $R$ is the set of free variables $\Yvar$ of $\phi$ for which the node belongs to $\nu(\Yvar)$, then
			$\rho=\Comp_{g(\eta),r,\phi}(R,\rho_1,\dots,\rho_r)$
			(there are only finitely many possibilities for $\rho,\rho_1,\dots,\rho_r\in\Pht\phi$, for $r\in\{0,\dots,\rMax\}$, for $\eta\in\Sigmaout(\Aa_\phi)$,
			and finitely many free variables of $\phi$, thus the constructed formula can be just a big alternative listing all possible cases), and
		\item	if a node with label $\eta\in\Sigmaout(\Aa_\phi)$ belongs to $\Xvar_\rho$ and none of $\nu(\Yvar)$ for $\Yvar$ free in $\phi$ contains this node or some its descendant,
			then $\rho=f_\phi(\eta)$.
		\end{itemize}
		Consider now a $(\Sigma,\rMax)$-tree $T$, and a valuation $\nu$ in this tree.
		If $\pht{\phi}{T}{\nu}=\tau$, then we can show that $\xi_{\phi,\tau}$ is true
		by taking for $\Xvar_\rho$ the set of nodes $\unode$ for which $\pht{\phi}{T\restr_\unode}{\nu\restr_\unode}=\rho$ (for every $\rho\in\Pht{\phi}$).
		Conversely, suppose that $\xi_{\phi,\tau}$ is true.
		Then we can prove that a node $\unode$ can belong to the set represented by $\Xvar_\rho$ (for $\rho\in\Pht{\phi}$) only when $\pht{\phi}{T\restr_\unode}{\nu\restr_\unode}=\rho$.
		The proof is by a straightforward induction on the number of descendants of $\unode$ that belong to $\nu(\Yvar)$ for some $\Yvar$ free in $\phi$;
		we use Lemma~\ref{lem:compositionality} for the induction step.
	\end{proof}

	Now the proof of Lemma~\ref{lem:logic-to-automata} follows easily.
	Indeed, when $\phi$ is a sentence (has no free variables), $\pht{\phi}{T}{\VALempty}$ determines whether $\phi$ holds in $T$.
	Thus, it is enough to take the automaton $\Bb_\phi$ constructed in Lemma~\ref{lem:logic-to-automata-aux},
	and replace the function $f_\phi$ by the set $\SigmaF=\{\eta\in\Sigmaout(\Aa)\mid\mathit{tv}_\phi(f_\phi(\eta))\}$
	(where $\mathit{tv}_\phi$, defined on page~\pageref{page:tv}, given a $\phi$-phenotype says whether $\phi$ holds in trees having this $\phi$-phenotype).

	We remark that the \wmsou logic (which is weaker than \msoufin) corresponds to nested \unbound-prefix automata, composed of \unbound-prefix automata only (i.e., not using \mso automata).
	Indeed, \mso automata are needed only to deal with infinite sets;
	when all quantified sets are finite, we can simulate all the constructs using \unbound-prefix automata~\cite{wmsou-schemes}.

	We also remark that Bojańczyk and Toruńczyk~\cite{wmso+u-trees} introduce another model of automata equivalent to \wmsou: nested limsup automata.
	A common property of these two models is that both of them are nested; the components of nested limsup automata are of a different form, though.

\section{Proof of the main theorem}\label{sec:main-thm}

	In this section we prove our main theorem---Theorem~\ref{thm:main}.
	To this end, we have to recall three properties of recursion schemes: logical reflection (Fact~\ref{fact:mso-reflection}), SUP reflection (Fact~\ref{fact:sup-reflection}),
	and closure under composition with finite tree transducers (Fact~\ref{fact:transducer}).

	The property of logical reflection for schemes comes from Broadbent, Carayol, Ong, and Serre~\cite{reflection}.
	They state it for sentences of $\mu$-calculus, but $\mu$-calculus and \mso are equivalent over infinite trees~\cite{EmersonJutla}.

	\begin{factfact}[logical reflection {\cite[Theorem 2(ii)]{reflection}}]\label{fact:mso-reflection}
		For every \mso sentence $\phi$ and every scheme $\Gg$ generating a tree $T$
		one can construct a scheme $\Gg_\phi$ that generates a tree of the same shape as $T$, and such that its every node $\unode$ is labeled by a pair $(a_\unode,b_\unode)$,
		where $a_\unode$ is the label of $\unode$ in $T$, and $b_\unode$ is $\true$ if $\phi$ is satisfied in $T\restr_\unode$ and $\false$ otherwise.
	\newqed\end{factfact}

	The SUP reflection is the heart of our proof.
	In order to talk about this property, we need a few more definitions.
	By $\#_a(U)$ we denote the number of $a$-labeled nodes in a (finite) tree $U$.
	For a set of (finite) trees $\Ll$ and a set of symbols $A$, we define a predicate $\SUP_A(\Ll)$, which holds if for every $n\in\Nat$ there is some $U_n\in \Ll$ such that for all $a\in A$ it holds that $\#_a(U_n)\geq n$.

	Originally, in the simultaneous unboundedness problem we consider nondeterministic higher-order recursion schemes, which instead of generating a single infinite tree, recognize a set of finite trees.
	We use here an equivalent formulation, in which the set of finite trees is encoded in a single infinite tree.
	To this end, we use a special letter $\nd$, denoting a nondeterministic choice.
	We write $T\rednd U$ if $U$ is obtained from $T$ by choosing some $\nd$-labeled node $\unode$ and some its child $\vnode$, and attaching $T\restr_\vnode$ in place of $T\restr_\unode$.
	In other words, $\rednd$ is the smallest relation such that $\symb{\nd}[T_1][\dots][T_\nChld]\rednd T_j$ for $j\in\{1,\dots,\nChld\}$,
	and if $T_j\rednd T_j'$ for some $j\in\{1,\dots,\nChld\}$, and $T_i=T_i'$ for all $i\in\{1,\dots,\nChld\}\setminus\{j\}$,
	then $\symb{a}[T_1][\dots][T_\nChld]\rednd \symb{a}[T_1'][\dots][T_\nChld']$.
	For a tree $T$, $\Ll(T)$ is the set of all finite trees $U$ such that $\#_\nd(U)=\#_\omega(U)=0$ and $T\rednd^* U$.

	\begin{factfact}[SUP reflection~{\cite[Theorem 10.1]{types-diagonal-journal}}]\label{fact:sup-reflection}
		For every scheme $\Gg$ generating a tree $T$ one can construct a scheme $\Gg_\mathit{SUP}$ that generates a tree of the same shape as $T$,
		and such that its every node $\unode$, having in $T$ label $a_\unode$, is labeled by
		\begin{itemize}
		\item	a pair $(a_\unode,\{A\subseteq\Sigma_\Gg\mid \SUP_A(\Ll(T\restr_\unode))\})$, if $a_\unode\neq\nd$, and
		\item	the letter $\nd$, if $a_\unode=\nd$.
		\newqed\end{itemize}
	\end{factfact}

	The third recalled fact (Fact~\ref{fact:transducer}) talks about finite tree transducers.
	A \emph{(deterministic, top-down) finite tree transducer} is a tuple $\Tt=(\Sigma,\rMax,Q,\qIni,\delta)$, where
	$\Sigma$ is a finite alphabet, $\rMax$ is the maximal arity of considered trees, $Q$ is a finite set of states, $\qIni\in Q$ is an initial state,
	and $\delta$ is a transition function mapping $Q\times\Sigma\times\{0,\dots,\rMax\}$ to finite lambda-terms.
	A triple $(q,a,\nChld)$ should be mapped by $\delta$ to a term that uses only node constructors and variables of the form $\xVar_{i,p}$, where $i\in\{1,\dots,\nChld\}$ and $p\in Q$ (applications and lambda-binders are not allowed);
	at least one node constructor has to be used (the whole $\delta(q,a,\nChld)$ cannot be equal to a variable).

	For a $(\Sigma,\rMax)$-tree $T$ and a state $q\in Q$, we define $\Tt_q(T)$ by coinduction, as follows:
	if $T=\symb{a}[T_1][\dots][T_\nChld]$, then $\Tt_q(T)$ is the tree obtained from $\delta(q,a,\nChld)$ by substituting $\Tt_p(T_i)$ for the variable $\xVar_{i,p}$, for all $i\in\{1,\dots,\nChld\}$ and $p\in Q$.
	In the root we start from the initial state, that is, we define $\Tt(T)=\Tt_{\qIni}(T)$.
	We have the following fact.

	\begin{factfact}\label{fact:transducer}
		For every finite tree transducer $\Tt=(\Sigma,\rMax,Q,\qIni,\delta)$ and every scheme $\Gg$ generating a $(\Sigma,\rMax)$-tree $T$,
		one can construct a scheme $\Gg_\Tt$ that generates the tree $\Tt(T)$.
	\newqed\end{factfact}

	This fact follows from the equivalence between schemes and collapsible pushdown systems~\cite{collapsible}, as it is straightforward to compose a collapsible pushdown system with $\Tt$
	(where due to Fact~\ref{fact:convergent} we can assume that $\Lambda(\Gg)$ is fully convergent, i.e., that every node of $T$ is explicitly generated by the collapsible pushdown system).
	Since we are not aware of any proof of this fact in the literature, we give more details in Appendix~\ref{app:transducer}.

	Using Fact~\ref{fact:mso-reflection} we can compose schemes with \mso automata, as stated below.

	\begin{lem}\label{lem:mso-and-schemes}
		For every \mso automaton $\Aa$ and every scheme $\Gg$ generating a tree $T$, where $\Sigma_\Gg\subseteq\Sigmain(\Aa)$, one can construct a scheme $\Gg_\Aa$ that generates the tree $\Aa(T)$.
	\end{lem}

	\begin{proof}
		Let $\Aa=(\Sigma,Q,(\phi_q)_{q\in Q})$.
		Assume that $Q=\{1,\dots,n\}$, and take $\Gg_0=\Gg$ and $T_0=T$.
		Consecutively for $q=1,\dots,n$ we want to apply Fact~\ref{fact:mso-reflection} to $\phi_q$ and $\Gg_{q-1}$, and obtain a scheme $\Gg_q$ that generates a tree $T_q$ of the same shape as $T$,
		and such that its every node $\unode$ is labeled by a tuple $(a,b_1,\dots,b_q)$, where $a$ is the label of $\unode$ in $T$,
		and $b_i$ says whether $\phi_i$ is satisfied in $T\restr_\unode$ for $i\in\{1,\dots,q\}$.
		Strictly speaking, we cannot apply Fact~\ref{fact:mso-reflection} to the original sentences $\phi_q$ (these sentences can be evaluated in $T$, but not in $T_{q-1}$).
		We need to slightly modify the sentences: out of $\phi_q$ we obtain $\phi_q'$ by changing every subformula of the form $a(\Xvar)$
		to a formula saying that every node in the set represented by $\Xvar$ is labeled by a letter from $\{a\}\times\{\true,\false\}^{q-1}$.
		Then $\phi_q$ is satisfied in $T\restr_\unode$ if and only if $\phi_q'$ is satisfied in $T_{q-1}\restr_\unode$;
		in consequence, we can apply Fact~\ref{fact:mso-reflection} to $\phi_q'$ and $\Gg_{q-1}$.

		The last tree, $T_n$, contains truth values of all sentences $\phi_q$.
		In order to obtain $\Gg_A$ as required, it is thus enough to rename letters appearing in $\Gg_n$:
		we change every letter $(a,b_1,\dots,b_n)$ to $(a,f)$ for $f\colon Q\to\{0,1,2\}$ mapping every $q\in Q$ to $1$ if $b_q=\true$, and to $0$ if $b_q=\false$.
	\end{proof}

	As one can expect, we can also compose schemes with \unbound-prefix automata, and for that we need Facts~\ref{fact:sup-reflection} and~\ref{fact:transducer}.

	\begin{lem}\label{lem:main-technical}
		For every \unbound-prefix automaton $\Aa$ and every scheme $\Gg$ generating a tree $T$, where $\Sigma_\Gg\subseteq\Sigmain(\Aa)$,
		one can construct a scheme $\Gg_\Aa$ that generates the tree $\Aa(T)$.
	\end{lem}

	It is easy to deduce Theorem~\ref{thm:main} out of Lemmata~\ref{lem:mso-and-schemes} and~\ref{lem:main-technical}.
	Indeed, consider an \msoufin sentence $\phi$ and a scheme $\Gg_0$ generating a tree $T_0$.
	By Lemma~\ref{lem:logic-to-automata}, $\phi$ is equivalent to a nested \unbound-prefix \mso automaton $\Aa=\Aa_1\vartriangleright\dots\vartriangleright\Aa_k$, together with an accepting set $\SigmaF$.
	By consecutively applying Lemmata~\ref{lem:mso-and-schemes} and~\ref{lem:main-technical} for $i=1,\dots,k$, we combine $\Gg_{i-1}$ with $\Aa_i$, obtaining a scheme $\Gg_i$ that generates the tree $T_i=\Aa_i(T_{i-1})$.
	The root of $T_k=\Aa(T_0)$ has label in $\SigmaF$ if and only if $\phi$ is satisfied in $T_0$.
	Surely this label can be read: having $\Gg_k$, we simply start generating the tree $T_k$, until its root is generated
	(by Fact~\ref{fact:convergent}, we can assume that $\Lambda(\Gg_k)$ is fully convergent).

	We now come to the proof of Lemma~\ref{lem:main-technical}.
	We are thus given a \unbound-prefix automaton $\Aa=(\Sigma,Q,\QImp,\Delta)$, and a scheme $\Gg$ generating a tree $T$, where $\Sigma_\Gg\subseteq\Sigma$;
	our goal is to create a scheme $\Gg_\Aa$ that generates the tree $\Aa(T)$.
	As a first step, we create a finite tree transducer $\Tt$ that converts $T$ into a tree containing all runs of $\Aa$ on all subtrees of $T$.
	Let us write $Q=\{p_1,\dots,p_{|Q|}\}$.
	As $\Tt$ we take $(\Sigma_\Gg,\rMax(\Gg),Q\cup\{\qIni,\top\},\qIni,\delta)$, where $\qIni\not\in Q$ is a fresh state, and $\delta$ is defined as follows.
	For $q\in Q$, $a\in\Sigma_\Gg$, and $\nChld\leq\rMax(\Gg)$ we take
	\begin{align*}
		\delta(q,a,\nChld)=\symb{\nd}[\symb{q}[\xVar_{1,q_{11}}][\dots][\xVar_{\nChld,q_{1\nChld}}]][\dots][\symb{q}[\xVar_{1,q_{k1}}][\dots][\xVar_{\nChld,q_{k\nChld}}]]\,,
	\end{align*}
	where $(q,a,q_{11},\dots,q_{1\nChld}),\dots,(q,a,q_{k1},\dots,q_{k\nChld})$ are all elements of $\Delta$ being of length $\nChld+2$ and having $q$ and $a$ on the first two coordinates.
	Moreover, for $a\in\Sigma_\Gg$ and $\nChld\leq\rMax(\Gg)$ (and for a special letter ``?'') we take
	\begin{align*}
		\delta(\qIni,a,\nChld)&=\symb{a}[\xVar_{1,\qIni}][\dots][\xVar_{\nChld,\qIni}][\symb{?}[\delta(p_1,a,\nChld)]][\dots][\symb{?}[\delta(p_{|Q|},a,\nChld)]]\,,
		&&\mbox{and}\\
		\delta(\top,a,\nChld)&=\symb{\top}\,.
	\end{align*}

	We see that $\Tt(T)$ contains all nodes of the original tree $T$.
	Additionally, below every node $\unode$ coming from $T$ we have $|Q|$ new children labeled by $?$, such that subtrees starting below these children describe runs of $\Aa$ on $T\restr_\unode$,
	starting in particular states.
	More precisely, when $\unode$ has $\nChld$ children in $T$, for every $i\in\{1,\dots,|Q|\}$
	there is a bijection between trees $U$ in $\Ll(\Tt(T)\restr_{\unode(\nChld+i)1})$ and runs $\rho$ of $\Aa$ on $T\restr_\unode$ such that $\rho(\epsilon)=p_i$.
	The label of every node $\vnode$ in such a tree $U$ contains the state assigned by $\rho$ to $\vnode$,
	where $U$ contains exactly all nodes to which $\rho$ assigns a state from $Q$, and all minimal nodes to which $\rho$ assigns $\top$
	(i.e., such that $\rho$ does not assign $\top$ to their parents).
	Recall that by definition $\rho$ can assign a state from $Q$ only to a finite prefix of the tree $T\restr_\unode$,
	which corresponds to the fact that $\Ll(\Tt(T)\restr_{\unode(\nChld+i)1})$ contains only finite trees.

	Actually, we need to consider a transducer $\Tt'$ obtained from $\Tt$ by a slight modification: we replace the letter $q$ appearing in $\delta(q,a,\nChld)$ by $1$ if $q\in\QImp$, and by $0$ if $q\not\in\QImp$.
	Then, for a node $\unode$ of $T$ having $\nChld$ children, and for $i\in\{1,\dots,|Q|\}$, we have the following equivalence:
	$\SUP_{\{1\}}(\Tt'(T)\restr_{\unode(\nChld+i)})$ holds if and only if for every $n\in\Nat$ there is a run $\rho_n$ of $\Aa$ on $T\restr_\unode$ that assigns $p_i$ to the root of $T\restr_\unode$,
	and such that for at least $n$ nodes $\vnode$ it holds that $\rho_n(\vnode)\in \QImp$.

	We now apply Fact~\ref{fact:transducer} to $\Gg$ and $\Tt'$; we obtain a scheme $\Gg_{\Tt'}$ that generates the tree $\Tt'(T)$.
	Then, we apply Fact~\ref{fact:sup-reflection} (SUP reflection) to $\Gg_{\Tt'}$, which gives us a scheme $\Gg'$.
	The tree $T'$ generated by $\Gg'$ has the same shape as $\Tt'(T)$, but in the label of every node $\vnode$ (originally having label other than $\nd$)
	there is additionally written a set $\Uu$ containing these sets $A\subseteq\Sigma_{\Gg_{\Tt'}}$ for which $\SUP_A(\Ll(T\restr_\vnode))$ holds.
	Next, using Fact~\ref{fact:mso-reflection} (logical reflection) $2|Q|$ times, we annotate every node $\unode$ of $T'$, having $\nChld'$ children, by logical values of the following properties, for $i=1,\dots,|Q|$:
	\begin{itemize}
	\item	whether $\nChld'\geq|Q|$ and $\Ll(T'\restr_{\unode(\nChld'-|Q|+i)1})$ is nonempty, and
	\item	whether $\nChld'\geq|Q|$ and the label $(a,\Uu)$ of node $\unode(\nChld'-|Q|+i)$ in $T'$ satisfies $\{1\}\in\Uu$.
	\end{itemize}
	Clearly both these properties can be expressed in \mso.
	For nodes $\unode$ coming from $T$, the first property holds when there is a run of $\Aa$ on $T\restr_\unode$ that assigns $p_i$ to the root of $T\restr_\unode$,
	and the second property holds when for every $n\in\Nat$ there is a run $\rho_n$ of $\Aa$ on $T\restr_\unode$ that assigns $p_i$ to the root of $T\restr_\unode$, and such that for at least $n$ nodes $\wnode$ it holds that $\rho_n(\wnode)\in \QImp$.
	Let $\Gg''$ be the scheme generating the tree $T''$ containing these annotations.

	Finally, we create $\Gg_\Aa$ by slightly modifying $\Gg''$:
	we replace every node constructor $\symb{(a,\Uu,\sigma_1,\tau_1,\dots,\sigma_{|Q|},\tau_{|Q|})}[P_1][\dots][P_{\nChld+|Q|}]$ with $\symb{(a,f)}[P_1][\dots][P_\nChld]$, where $f\colon Q\to\{0,1,2\}$ is such that
	$f(p_i)=2$ if $\tau_i=\true$, and $f(p_i)=1$ if $\sigma_i=\true$ but $\tau_i=\false$, and $f(p_i)=0$ otherwise, for all $i\in\{1,\dots,|Q|\}$
	(we do not do anything with node constructors of arity smaller than $|Q|$).
	As a result, only the nodes coming from $T$ remain, and they are appropriately relabeled.

\section{Extensions}\label{sec:conclusion}

	In this section we give a few possible extensions of our main theorem, saying that we can evaluate \msoufin sentences on trees generated by recursion schemes.
	First, we notice that our solution actually proves a stronger result: logical reflection for \msoufin.

	\begin{thm}\label{thm:msoufin-reflection}
		For every \msoufin sentence $\phi$ and every scheme $\Gg$ generating a tree $T$ one can construct a scheme $\Gg_\phi$ that generates a tree of the same shape as $T$,
		and such that its every node $\unode$ is labeled by a pair $(a_\unode,b_\unode)$,
		where $a_\unode$ is the label of $\unode$ in $T$, and $b_\unode$ is $\true$ if $\phi$ is satisfied in $T\restr_\unode$ and $\false$ otherwise.
	\end{thm}

	\begin{proof}
		In the proof of Theorem~\ref{thm:main} we have constructed a nested \unbound-prefix automaton $\Aa$ equivalent to $\phi$, and then a scheme $\Gg_\Aa$ that generates the tree $\Aa(T)$.
		In every node $\unode$ of $\Aa(T)$ it is written whether $T\restr_\unode$ satisfies $\phi$.
		Moreover, labels of $\Aa(T)$ contain also original labels coming from $T$.
		Thus in order to obtain $\Gg_\phi$ it is enough to appropriately relabel node constructors appearing in $\Gg_\Aa$.
	\end{proof}

	In Theorem~\ref{thm:msoufin-reflection}, the sentence $\phi$ talks only about the subtree starting in $\unode$.
	One can obtain a stronger version of logical reflection (Theorem~\ref{thm:better-reflection}), where $\phi$ is a formula allowed to talk about $\unode$ in the context of the whole tree.
	This version can be obtained as a simple corollary of Theorem~\ref{thm:msoufin-reflection} by using the same methods as in Broadbent et al.~\cite[Proof of Corollary 2]{reflection}.
	As shown on page~\pageref{proof:better-reflection}, it is also an immediate consequence of our next theorem (Theorem~\ref{thm:msoufin-selection}).

	\begin{thm}\label{thm:better-reflection}
		For every \msoufin formula $\phi(\Xvar)$ with one free variable $\Xvar$ and every scheme $\Gg$ generating a tree $T$,
		one can construct a scheme $\Gg_\phi$ that generates a tree of the same shape as $T$, and such that its every node $\unode$ is labeled by a pair $(a_\unode,b_\unode)$,
		where $a_\unode$ is the label of $\unode$ in $T$, and $b_\unode$ is $\true$ if $\phi$ is satisfied in $T$ with $\Xvar$ valuated to $\{\unode\}$, and $\false$ otherwise.
	\end{thm}

	Carayol and Serre~\cite{selection} show one more property of \mso and schemes, called effective selection.
	This time we are given an \mso sentence $\psi$ of the form $\exists \Xvar.\phi$.
	Assuming that $\psi$ is satisfied in the tree $T$ generated by a scheme $\Gg$, one wants to compute an example set $X$ of nodes of $T$,
	such that $\phi$ is true in $T$ with the variable $\Xvar$ valuated to this set $X$.
	The theorem says that it is possible to create a scheme $\Gg_\phi$ that generates a tree of the same shape as $T$, in which nodes belonging to some such example set $X$ are marked.
	We can show the same for \msoufin.

	\begin{thm}\label{thm:msoufin-selection}
		For every \msoufin formula $\phi(\Xvar)$ with one free variable $\Xvar\in\Vv^\mathsf{inf}$ and every scheme $\Gg$ generating a tree $T$,
		if $\exists\Xvar.\phi(\Xvar)$ holds in $T$, then
		one can construct a scheme $\Gg_\phi$ that generates a tree $T'$ of the same shape as $T$, and such that its every node $\unode$ is labeled by a pair $(a_\unode,b_\unode)$,
		where $a_\unode$ is the label of $\unode$ in $T$, and $b_\unode$ belongs to $\{\true,\false\}$; when $X$ is the set of nodes of $T'$ having $\true$ on the second coordinate of the label,
		$\phi$ is holds in $T$ with $\Xvar$ valuated to $X$.
	\end{thm}

	The proof of this theorem bases on the following lemma, which is also interesting in itself.

	\begin{lem}\label{lem:selection-aux}
		For every \msoufin formula $\phi$ and every scheme $\Gg$ generating a tree $T$ one can construct a scheme $\Gg_+$ that generates a tree $T'$ of the same shape as $T$,
		and an \mso formula $\phi_\mathit{MSO}$ (whose all free variables are also free in $\phi$) such that for every valuation $\nu$ in $T$ (defined at least for all free variables of $\phi$)
		it holds that $T',\nu\models\phi_\mathit{MSO}$ if and only if $T,\nu\models\phi$.
		Moreover, the label of every node of $T'$ contains as its part the label of that node in $T$.
	\end{lem}

	\begin{proof}
		Recall that Lemma~\ref{lem:logic-to-automata-aux} gives us a nested \unbound-prefix \mso automaton $\Aa_\phi$ and \mso formulae $\xi_{\phi,\tau}$ for all $\tau\in\Pht{\phi}$
		such that for every valuation $\nu$ in $T$ (where $T$ is now the fixed $(\Sigma_\Gg,\rMax(\Gg))$-tree generated by $\Gg$)
		it holds that $\Aa_\phi(T),\nu\models\xi_{\phi,\tau}$ if and only if $\pht{\phi}{T}{\nu}=\tau$.

		Applying Lemmata~\ref{lem:mso-and-schemes} and~\ref{lem:main-technical} to components of the automaton $\Aa_\phi$,
		out of the scheme $\Gg$ we can construct a scheme $\Gg_+$ that generates the tree $\Aa_\phi(T)$.

		Recall that $\mathit{tv}_\phi(\tau)$ says whether $\phi$ is true in a tree having $\phi$-phenotype $\tau$, and
		consider the \mso formula
		\begin{align*}
			\phi_\mathit{MSO}\equiv\bigvee_{\substack{\tau\in\Pht\phi\\\mathit{tv}_\phi(\tau)}}\xi_{\phi,\tau}\,.
		\end{align*}
		By the above, for every valuation $\nu$, it holds that $\Aa_\phi(T),\nu\models\phi_\mathit{MSO}$ if and only if $T,\nu\models\phi$, as required in the thesis.
	\end{proof}

	Using the above lemma, we can easily deduce effective selection for \msoufin out of effective selection for \mso.

	\begin{proof}[Proof of Theorem~\ref{thm:msoufin-selection}]
		Using effective selection for \mso (which is a theorem with the same statement as Theorem~\ref{thm:msoufin-selection}, but for the \mso logic~\cite{selection})
		for the formula $\phi_\mathit{MSO}$ and for the scheme $\Gg_+$ (created by Lemma~\ref{lem:selection-aux}) we obtain a scheme $\Gg_\phi'$.
		It is almost as required: it generates a tree $T'$ of the same shape as $T$ (but with some additional parts of labels, added by $\Gg_+$),
		where additionally nodes of some set $X$ are marked, so that $\phi_\mathit{MSO}$ holds in $T'$ with $\Xvar$ valuated to $X$.
		Lemma~\ref{lem:selection-aux} implies that then also $\phi$ holds in $T$ with $\Xvar$ valuated to $X$.
		Thus, it is enough to modify node constructors of $\Gg_\phi'$: out of every letter we leave only the original letter coming from $\Gg$, and the last component marking the set $X$,
		while we remove all the components added by $\Gg_+$.
	\end{proof}

	One may want to obtain an analogous theorem for $\Xvar\in\Vv^\mathsf{fin}$, that is, for a sentence of the form $\efin\Xvar.\phi(X)$.
	It is, however, a special case of Theorem~\ref{thm:msoufin-selection}, which can be used with the sentence $\exists\Xvar'.\efin\Xvar.\Xvar\subseteq\Xvar'\land\Xvar'\subseteq\Xvar\land\phi(\Xvar)$.
	We remark, though, that the version of Theorem~\ref{thm:msoufin-selection} for $\Xvar\in\Vv^\mathsf{fin}$ is actually also a corollary of Theorem~\ref{thm:main},
	because there are only countably many finite sets $X$, so we may try one after another, until we find some set for which $\phi$ is satisfied;
	it is easy to hardcode a given set $X$ in the formula (or in the scheme).

	We now show how Theorem~\ref{thm:better-reflection} follows from Theorem~\ref{thm:msoufin-selection}.

	\begin{proof}[Proof of Theorem~\ref{thm:better-reflection}]\label{proof:better-reflection}
		We use Theorem~\ref{thm:msoufin-selection} for
		\begin{align*}
			\phi'(\Xvar')\equiv\forall\Xvar.\mathit{sing}(\Xvar)\arr(\Xvar\subseteq\Xvar'\arr\phi(\Xvar))\land(\neg(\Xvar\subseteq\Xvar')\arr\neg\phi(\Xvar))\,.
		\end{align*}
		The only set $X'$ for which $\phi'$ is true in a tree $T$ is the set of all nodes $\unode$ for which $\phi$ is true in $T$ with $\Xvar$ valuated to $\{\unode\}$.
		Thus the scheme $\Gg_{\phi'}$ obtained from Theorem~\ref{thm:msoufin-selection} satisfies the thesis of Theorem~\ref{thm:better-reflection}.
	\end{proof}

	Our algorithm for Theorem~\ref{thm:main} has nonelementary complexity.
	This is unavoidable, as already model-checking of WMSO sentences on the infinite word over an unary alphabet is nonelementary.
	It would be interesting to find some other formalism for expressing unboundedness properties, maybe using some model of automata, for which the model-checking problem has better complexity.
	We leave this issue for future work.

	Finally, we remark that in our solution we do not use the full power of the simultaneous unboundedness problem, we only use the single-letter case.
	On the other hand, it seems that \msoufin is not capable to express simultaneous unboundedness, only its single-letter case.
	Thus, another direction for a future work is to extend \msoufin to a logic that can actually express simultaneous unboundedness.
	As a possible candidate we see the qcMSO logic introduced in Kaiser, Lang, Le{\ss}enich, and L{\"{o}}ding~\cite{qcMSO}, in which simultaneous unboundedness is expressible.

\bibliographystyle{alpha}
\bibliography{bib}

\begin{thebibliography}{CMvRZ15}

\bibitem[Aho68]{AhoIndexed}
Alfred~V. Aho.
\newblock Indexed grammars - an extension of context-free grammars.
\newblock {\em J. {ACM}}, 15(4):647--671, 1968.

\bibitem[BCCC96]{OrderedMultiPushdown}
Luca Breveglieri, Alessandra Cherubini, Claudio Citrini, and Stefano
  Crespi{-}Reghizzi.
\newblock Multi-push-down languages and grammars.
\newblock {\em Int. J. Found. Comput. Sci.}, 7(3):253--292, 1996.

\bibitem[BCL08]{BlumensathColcombet}
Achim Blumensath, Thomas Colcombet, and Christof L{\"{o}}ding.
\newblock Logical theories and compatible operations.
\newblock In J{\"{o}}rg Flum, Erich Gr{\"{a}}del, and Thomas Wilke, editors,
  {\em Logic and Automata: History and Perspectives [in Honor of Wolfgang
  Thomas]}, volume~2 of {\em Texts in Logic and Games}, pages 73--106.
  Amsterdam University Press, 2008.

\bibitem[BCOS10]{reflection}
Christopher~H. Broadbent, Arnaud Carayol, C.{-}H.~Luke Ong, and Olivier Serre.
\newblock Recursion schemes and logical reflection.
\newblock In {\em Proceedings of the 25th Annual {IEEE} Symposium on Logic in
  Computer Science, {LICS} 2010, 11-14 July 2010, Edinburgh, United Kingdom},
  pages 120--129. {IEEE} Computer Society, 2010.

\bibitem[BK13]{HorSat}
Christopher~H. Broadbent and Naoki Kobayashi.
\newblock Saturation-based model checking of higher-order recursion schemes.
\newblock In Simona Ronchi~Della Rocca, editor, {\em Computer Science Logic
  2013 {(CSL} 2013), {CSL} 2013, September 2-5, 2013, Torino, Italy}, volume~23
  of {\em LIPIcs}, pages 129--148. Schloss Dagstuhl - Leibniz-Zentrum fuer
  Informatik, 2013.

\bibitem[BO09]{globalMC}
Christopher~H. Broadbent and C.{-}H.~Luke Ong.
\newblock On global model checking trees generated by higher-order recursion
  schemes.
\newblock In Luca de~Alfaro, editor, {\em Foundations of Software Science and
  Computational Structures, 12th International Conference, {FOSSACS} 2009, Held
  as Part of the Joint European Conferences on Theory and Practice of Software,
  {ETAPS} 2009, York, UK, March 22-29, 2009. Proceedings}, volume 5504 of {\em
  Lecture Notes in Computer Science}, pages 107--121. Springer, 2009.

\bibitem[Boj04]{BojanczykU}
Miko{\l}aj Boja{\'{n}}czyk.
\newblock A bounding quantifier.
\newblock In Jerzy Marcinkowski and Andrzej Tarlecki, editors, {\em Computer
  Science Logic, 18th International Workshop, {CSL} 2004, 13th Annual
  Conference of the EACSL, Karpacz, Poland, September 20-24, 2004,
  Proceedings}, volume 3210 of {\em Lecture Notes in Computer Science}, pages
  41--55. Springer, 2004.

\bibitem[Boj11]{wmso+u-words}
Miko{\l}aj Boja{\'{n}}czyk.
\newblock Weak {MSO} with the unbounding quantifier.
\newblock {\em Theory Comput. Syst.}, 48(3):554--576, 2011.

\bibitem[Boj14]{wmso+up}
Miko{\l}aj Boja{\'{n}}czyk.
\newblock Weak {MSO+U} with path quantifiers over infinite trees.
\newblock In Javier Esparza, Pierre Fraigniaud, Thore Husfeldt, and Elias
  Koutsoupias, editors, {\em Automata, Languages, and Programming - 41st
  International Colloquium, {ICALP} 2014, Copenhagen, Denmark, July 8-11, 2014,
  Proceedings, Part {II}}, volume 8573 of {\em Lecture Notes in Computer
  Science}, pages 38--49. Springer, 2014.

\bibitem[BPT16]{mso+u-undecid}
Miko{\l}aj Boja{\'{n}}czyk, Paweł Parys, and Szymon Toru{\'{n}}czyk.
\newblock The {MSO+U} theory of ({N}, {\textless}) is undecidable.
\newblock In Nicolas Ollinger and Heribert Vollmer, editors, {\em 33rd
  Symposium on Theoretical Aspects of Computer Science, {STACS} 2016, February
  17-20, 2016, Orl{\'{e}}ans, France}, volume~47 of {\em LIPIcs}, pages
  21:1--21:8. Schloss Dagstuhl - Leibniz-Zentrum fuer Informatik, 2016.

\bibitem[BT12]{wmso+u-trees}
Miko{\l}aj Boja{\'{n}}czyk and Szymon Toru{\'{n}}czyk.
\newblock Weak {MSO+U} over infinite trees.
\newblock In Christoph D{\"{u}}rr and Thomas Wilke, editors, {\em 29th
  International Symposium on Theoretical Aspects of Computer Science, {STACS}
  2012, February 29th - March 3rd, 2012, Paris, France}, volume~14 of {\em
  LIPIcs}, pages 648--660. Schloss Dagstuhl - Leibniz-Zentrum fuer Informatik,
  2012.

\bibitem[CMvRZ15]{sep-piecewise-test}
Wojciech Czerwiński, Wim Martens, Lorijn van Rooijen, and Marc Zeitoun.
\newblock A note on decidable separability by piecewise testable languages.
\newblock In Adrian Kosowski and Igor Walukiewicz, editors, {\em Fundamentals
  of Computation Theory - 20th International Symposium, {FCT} 2015,
  Gda{\'{n}}sk, Poland, August 17-19, 2015, Proceedings}, volume 9210 of {\em
  Lecture Notes in Computer Science}, pages 173--185. Springer, 2015.

\bibitem[CPSW15]{Ordered-Tree-Pushdown}
Lorenzo Clemente, Paweł Parys, Sylvain Salvati, and Igor Walukiewicz.
\newblock Ordered tree-pushdown systems.
\newblock In Prahladh Harsha and G.~Ramalingam, editors, {\em 35th {IARCS}
  Annual Conference on Foundation of Software Technology and Theoretical
  Computer Science, {FSTTCS} 2015, December 16-18, 2015, Bangalore, India},
  volume~45 of {\em LIPIcs}, pages 163--177. Schloss Dagstuhl - Leibniz-Zentrum
  fuer Informatik, 2015.

\bibitem[CPSW16]{diagonal}
Lorenzo Clemente, Paweł Parys, Sylvain Salvati, and Igor Walukiewicz.
\newblock The diagonal problem for higher-order recursion schemes is decidable.
\newblock In Martin Grohe, Eric Koskinen, and Natarajan Shankar, editors, {\em
  Proceedings of the 31st Annual {ACM/IEEE} Symposium on Logic in Computer
  Science, {LICS} '16, New York, NY, USA, July 5-8, 2016}, pages 96--105.
  {ACM}, 2016.

\bibitem[CS12]{selection}
Arnaud Carayol and Olivier Serre.
\newblock Collapsible pushdown automata and labeled recursion schemes:
  Equivalence, safety and effective selection.
\newblock In {\em Proceedings of the 27th Annual {IEEE} Symposium on Logic in
  Computer Science, {LICS} 2012, Dubrovnik, Croatia, June 25-28, 2012}, pages
  165--174. {IEEE} Computer Society, 2012.

\bibitem[Dam82]{Damm82}
Werner Damm.
\newblock The {IO-} and {OI}-hierarchies.
\newblock {\em Theor. Comput. Sci.}, 20:95--207, 1982.

\bibitem[EJ91]{EmersonJutla}
E.~Allen Emerson and Charanjit~S. Jutla.
\newblock Tree automata, mu-calculus and determinacy (extended abstract).
\newblock In {\em 32nd Annual Symposium on Foundations of Computer Science, San
  Juan, Puerto Rico, 1-4 October 1991}, pages 368--377. {IEEE} Computer
  Society, 1991.

\bibitem[FV59]{FefermanVaught}
Solomon Feferman and Robert~Lawson Vaught.
\newblock The first order properties of products of algebraic systems.
\newblock {\em Fundamenta Mathematicae}, 47(1):57--103, 1959.

\bibitem[GK10]{wmso+u-kaiser}
Tobias Ganzow and {\L}ukasz Kaiser.
\newblock New algorithm for weak monadic second-order logic on inductive
  structures.
\newblock In Anuj Dawar and Helmut Veith, editors, {\em Computer Science Logic,
  24th International Workshop, {CSL} 2010, 19th Annual Conference of the EACSL,
  Brno, Czech Republic, August 23-27, 2010. Proceedings}, volume 6247 of {\em
  Lecture Notes in Computer Science}, pages 366--380. Springer, 2010.

\bibitem[Had12]{haddad-fics}
Axel Haddad.
\newblock {IO} vs {OI} in higher-order recursion schemes.
\newblock In Dale Miller and Zolt{\'{a}}n {\'{E}}sik, editors, {\em Proceedings
  8th Workshop on Fixed Points in Computer Science, {FICS} 2012, Tallinn,
  Estonia, 24th March 2012.}, volume~77 of {\em {EPTCS}}, pages 23--30, 2012.

\bibitem[HKO16]{diagonal-safe}
Matthew Hague, Jonathan Kochems, and C.{-}H.~Luke Ong.
\newblock Unboundedness and downward closures of higher-order pushdown
  automata.
\newblock In Rastislav Bod{\'{\i}}k and Rupak Majumdar, editors, {\em
  Proceedings of the 43rd Annual {ACM} {SIGPLAN-SIGACT} Symposium on Principles
  of Programming Languages, {POPL} 2016, St. Petersburg, FL, USA, January 20 -
  22, 2016}, pages 151--163. {ACM}, 2016.

\bibitem[HM12]{hummel-skrzypczak-topological}
Szczepan Hummel and {Micha{\l} Skrzypczak}.
\newblock The topological complexity of {MSO+U} and related automata models.
\newblock {\em Fundam. Inform.}, 119(1):87--111, 2012.

\bibitem[HMOS08]{collapsible}
Matthew Hague, Andrzej~S. Murawski, C.{-}H.~Luke Ong, and Olivier Serre.
\newblock Collapsible pushdown automata and recursion schemes.
\newblock In {\em Proceedings of the Twenty-Third Annual {IEEE} Symposium on
  Logic in Computer Science, {LICS} 2008, 24-27 June 2008, Pittsburgh, PA,
  {USA}}, pages 452--461. {IEEE} Computer Society, 2008.

\bibitem[KLLL15]{qcMSO}
{\L}ukasz Kaiser, Martin Lang, Simon Le{\ss}enich, and Christof L{\"{o}}ding.
\newblock A unified approach to boundedness properties in {MSO}.
\newblock In Kreutzer \cite{DBLP:conf/csl/2015}, pages 441--456.

\bibitem[KNU02]{KNU-hopda}
Teodor Knapik, Damian Niwiński, and Paweł Urzyczyn.
\newblock Higher-order pushdown trees are easy.
\newblock In Mogens Nielsen and Uffe Engberg, editors, {\em Foundations of
  Software Science and Computation Structures, 5th International Conference,
  {FOSSACS} 2002. Held as Part of the Joint European Conferences on Theory and
  Practice of Software, {ETAPS} 2002 Grenoble, France, April 8-12, 2002,
  Proceedings}, volume 2303 of {\em Lecture Notes in Computer Science}, pages
  205--222. Springer, 2002.

\bibitem[KO09]{KobayashiOngtypes}
Naoki Kobayashi and C.{-}H.~Luke Ong.
\newblock A type system equivalent to the modal mu-calculus model checking of
  higher-order recursion schemes.
\newblock In {\em Proceedings of the 24th Annual {IEEE} Symposium on Logic in
  Computer Science, {LICS} 2009, 11-14 August 2009, Los Angeles, CA, {USA}},
  pages 179--188. {IEEE} Computer Society, 2009.

\bibitem[Kob11]{GTRecS}
Naoki Kobayashi.
\newblock A practical linear time algorithm for trivial automata model checking
  of higher-order recursion schemes.
\newblock In Martin Hofmann, editor, {\em Foundations of Software Science and
  Computational Structures - 14th International Conference, {FOSSACS} 2011,
  Held as Part of the Joint European Conferences on Theory and Practice of
  Software, {ETAPS} 2011, Saarbr{\"{u}}cken, Germany, March 26-April 3, 2011.
  Proceedings}, volume 6604 of {\em Lecture Notes in Computer Science}, pages
  260--274. Springer, 2011.

\bibitem[Kob13]{KobayashiPrograms}
Naoki Kobayashi.
\newblock Model checking higher-order programs.
\newblock {\em J. {ACM}}, 60(3):20:1--20:62, 2013.

\bibitem[Kre15]{DBLP:conf/csl/2015}
Stephan Kreutzer, editor.
\newblock {\em 24th {EACSL} Annual Conference on Computer Science Logic, {CSL}
  2015, September 7-10, 2015, Berlin, Germany}, volume~41 of {\em LIPIcs}.
  Schloss Dagstuhl - Leibniz-Zentrum fuer Informatik, 2015.

\bibitem[Lä68]{Lauchli}
Hans Läuchli.
\newblock A decision procedure for the weak second order theory of linear
  order.
\newblock {\em Studies in Logic and the Foundations of Mathematics},
  50:189--197, 1968.

\bibitem[NO14]{TravMC2}
Robin~P. Neatherway and C.{-}H.~Luke Ong.
\newblock {TravMC2}: Higher-order model checking for alternating parity tree
  automata.
\newblock In Neha Rungta and Oksana Tkachuk, editors, {\em 2014 International
  Symposium on Model Checking of Software, {SPIN} 2014, Proceedings, San Jose,
  CA, USA, July 21-23, 2014}, pages 129--132. {ACM}, 2014.

\bibitem[Ong06]{Ong-hoschemes}
C.{-}H.~Luke Ong.
\newblock On model-checking trees generated by higher-order recursion schemes.
\newblock In {\em 21th {IEEE} Symposium on Logic in Computer Science {(LICS}
  2006), 12-15 August 2006, Seattle, WA, USA, Proceedings}, pages 81--90.
  {IEEE} Computer Society, 2006.

\bibitem[Par17]{types-diagonal}
Paweł Parys.
\newblock The complexity of the diagonal problem for recursion schemes.
\newblock In Satya~V. Lokam and R.~Ramanujam, editors, {\em 37th {IARCS} Annual
  Conference on Foundations of Software Technology and Theoretical Computer
  Science, {FSTTCS} 2017, December 11-15, 2017, Kanpur, India}, volume~93 of
  {\em LIPIcs}, pages 45:1--45:14. Schloss Dagstuhl - Leibniz-Zentrum fuer
  Informatik, 2017.

\bibitem[Par18a]{wmsou-schemes}
Paweł Parys.
\newblock Recursion schemes and the {WMSO+U} logic.
\newblock In Rolf Niedermeier and Brigitte Vall{\'{e}}e, editors, {\em 35th
  Symposium on Theoretical Aspects of Computer Science, {STACS} 2018, February
  28 to March 3, 2018, Caen, France}, volume~96 of {\em LIPIcs}, pages
  53:1--53:16. Schloss Dagstuhl - Leibniz-Zentrum fuer Informatik, 2018.

\bibitem[Par18b]{types-diagonal-journal}
Paweł Parys.
\newblock A type system describing unboundedness.
\newblock Submitted, 2018.

\bibitem[PT16]{wmso-model}
Paweł Parys and Szymon Toru{\'{n}}czyk.
\newblock Models of lambda-calculus and the weak {MSO} logic.
\newblock In Jean{-}Marc Talbot and Laurent Regnier, editors, {\em 25th {EACSL}
  Annual Conference on Computer Science Logic, {CSL} 2016, August 29 -
  September 1, 2016, Marseille, France}, volume~62 of {\em LIPIcs}, pages
  11:1--11:12. Schloss Dagstuhl - Leibniz-Zentrum fuer Informatik, 2016.

\bibitem[RNO14]{PrefaceTool}
Steven~J. Ramsay, Robin~P. Neatherway, and C.{-}H.~Luke Ong.
\newblock A type-directed abstraction refinement approach to higher-order model
  checking.
\newblock In Suresh Jagannathan and Peter Sewell, editors, {\em The 41st Annual
  {ACM} {SIGPLAN-SIGACT} Symposium on Principles of Programming Languages,
  {POPL} '14, San Diego, CA, USA, January 20-21, 2014}, pages 61--72. {ACM},
  2014.

\bibitem[She75]{Shelah}
Saharon Shelah.
\newblock The monadic theory of order.
\newblock {\em Annals of Mathematics}, 102(3):379--419, 1975.

\bibitem[SW14]{KrivineWS}
Sylvain Salvati and Igor Walukiewicz.
\newblock Krivine machines and higher-order schemes.
\newblock {\em Inf. Comput.}, 239:340--355, 2014.

\bibitem[SW15a]{ModelSM}
Sylvain Salvati and Igor Walukiewicz.
\newblock A model for behavioural properties of higher-order programs.
\newblock In Kreutzer \cite{DBLP:conf/csl/2015}, pages 229--243.

\bibitem[SW15b]{models-reflection}
Sylvain Salvati and Igor Walukiewicz.
\newblock Using models to model-check recursive schemes.
\newblock {\em Logical Methods in Computer Science}, 11(2), 2015.

\bibitem[SW16]{schemes-lY}
Sylvain Salvati and Igor Walukiewicz.
\newblock Simply typed fixpoint calculus and collapsible pushdown automata.
\newblock {\em Mathematical Structures in Computer Science}, 26(7):1304--1350,
  2016.

\bibitem[Zet15]{Zetzsche-dc}
Georg Zetzsche.
\newblock An approach to computing downward closures.
\newblock In Magn{\'{u}}s~M. Halld{\'{o}}rsson, Kazuo Iwama, Naoki Kobayashi,
  and Bettina Speckmann, editors, {\em Automata, Languages, and Programming -
  42nd International Colloquium, {ICALP} 2015, Kyoto, Japan, July 6-10, 2015,
  Proceedings, Part {II}}, volume 9135 of {\em Lecture Notes in Computer
  Science}, pages 440--451. Springer, 2015.

\end{thebibliography}

\appendix

\section{Proof of Fact~\ref{fact:transducer}
		}\label{app:transducer}

	As already said, Fact~\ref{fact:transducer} follows easily from the equivalence between schemes and collapsible pushdown systems.
	We do not even need to know a full definition of these systems.
	Let us recall those fragments that are relevant for us.

	For every $n\in\Nat$, and every finite set $\Gamma$ containing a distinguished initial symbol $\bot\in\Gamma$, there are defined:
	\begin{itemize}
	\item	a set $\PD{n}{\Gamma}$ of collapsible pushdowns of order $n$ over stack alphabet $\Gamma$,
	\item	an initial pushdown $\bot_n\in\PD{n}{\Gamma}$,
	\item	a finite set $\Op{n}{\Gamma}$ of operations on these pushdowns, where every $\op\in\Op{n}{\Gamma}$ is a partial function from $\PD{n}{\Gamma}$ to $\PD{n}{\Gamma}$, and
	\item	a function $\topS\colon\PD{n}{\Gamma}\to\Gamma$ (returning the topmost symbol of a pushdown).
	\end{itemize}
	We assume that $\Op{n}{\Gamma}$ contains the identity operation $\id$, mapping every element of $\PD{n}{\Gamma}$ to itself.

	Having the above, we define a \emph{collapsible pushdown system} (a \emph{CPS} for short) as a tuple $\Cc=(Q,\qIni,n,\Gamma,\delta)$, where $Q$ is a set of states, $\qIni\in Q$ is an initial state,
	$n\in\Nat$ is an order, $\Gamma$ is a finite stack alphabet,
	and $\delta\colon Q\times\Gamma\to (Q\times\Op{n}{\Gamma})\uplus(\Sigma\times Q^*)$ is a transition function (where $\Sigma$ is some alphabet).
	A \emph{configuration} of $\Cc$ is a pair $(q,s)\in Q\times\PD{n}{\Gamma}$.
	A configuration $(p,t)$ is a \emph{successor} of $(q,s)$, written $(q,s)\to_\Cc(p,t)$, if $\delta(q,\topS(s))=(p,\op)$ and $\op(s)=t$.
	We define when a tree is generated by $\Cc$ from $(q,s)$, by coinduction:
	\begin{itemize}
	\item	if $(q,s)\to_\Cc^*(p,t)$, and $\delta(p,\topS(t))=(a,q_1,\dots,q_\nChld)\in\Sigma\times Q^*$, and trees $T_1,\dots,T_r$ are generated by $\Cc$ from $(q_1,t),\dots,(q_\nChld,t)$, respectively,
		then $\symb{a}[T_1][\dots][T_\nChld]$ is generated by $\Cc$ from $(q,s)$,
	\item	if there is no $(p,t)$ such that $(q,s)\to_\Cc^*(p,t)$ and $\delta(p,\topS(t))\in\Sigma\times Q^*$, then $\symb{\omega}$ is generated by $\Cc$ from $(q,s)$.
	\end{itemize}
	Notice that for every configuration $(q,s)$ there is at most one configuration $(p,t)$ such that $(q,s)\to_\Cc^*(p,t)$ and $\delta(p,\topS(t))\in\Sigma\times Q^*$;
	in consequence exactly one tree is generated by $\Cc$ from every configuration.
	While talking about the tree generated by $\Cc$, without referring to a configuration, we mean generating from the initial configuration $(\qIni,\bot_n)$.

	We say that a CPS is \emph{fully convergent} (from a configuration $(q,s)$) if it generates (from $(q,s)$) a tree without using the second item of the definition.
	More formally: we consider the CPS $\Cc_{-\omega}$ obtained from $\Cc$ by replacing $\omega$ with some other letter $\omega'$ (in all transitions), and
	we say that $\Cc$ is fully convergent (from $(q,s)$) if $\Cc_{-\omega}$ generates (from $(q,s)$) a tree without $\omega$-labeled nodes.
	We have the following fact.

	\begin{factfactC}[\cite{collapsible}]\label{fact:scheme-equiv-cps}
		For every scheme $\Gg$ one can construct a CPS $\Cc$ that generates the tree generated by $\Gg$ and, conversely, for every CPS $\Cc$ one can construct a scheme $\Gg$ that generates the tree generated by $\Cc$.
		Both translations preserve the property of being fully convergent.\footnote{%
			Clearly only a fully-convergent CPS/scheme can generate a tree without $\omega$-labeled nodes.
			Thus it is easy to preserve the property of being fully convergent: we can replace all appearances of $\omega$ by some fresh letter $\omega'$, switch to the other formalism,
			and then replace $\omega'$ back by $\omega$.}
	\newqed\end{factfactC}

	In Fact~\ref{fact:transducer} we are given a finite tree transducer $\Tt=(\Sigma,\rMax,P,\pIni,\delta_\Tt)$, and a scheme $\Gg$ generating a $(\Sigma,\rMax)$-tree $T$,
	and we want to construct a scheme $\Gg_\Tt$ that generates the tree $\Tt(T)$.
	By Fact~\ref{fact:convergent} we can assume that $\Gg$ is fully convergent.
	As a first step, we translate $\Gg$ to a fully convergent CPS $\Cc=(Q,\qIni,n,\Gamma,\delta_\Cc)$ generating $T$, using Fact~\ref{fact:scheme-equiv-cps}.

	Then, we create a CPS $\Cc_\Tt=(R,(\qIni,\pIni),n,\Gamma,\delta)$ by combining $\Cc$ with $\Tt$.
	Its set of states $R$ contains states of two kinds: pairs $(q,p)\in Q\times P$, and pairs $(q,U)$
	where $q\in Q$ and $U$ is a subterm of $\delta_\Tt(p,a,\nChld)$ for some $(p,a,\nChld)\in P\times\Sigma\times\{0,\dots,\rMax\}$.
	We define the transitions as follows:
	\begin{itemize}
	\item	if $\delta_\Cc(q,\chi)=(q',\op)\in Q\times\Op{n}{\Gamma}$, then $\delta((q,p),\chi)=((q',p),\op)$,
	\item	if $\delta_\Cc(q,\chi)=(a,q_1,\dots,q_\nChld)\in\Sigma\times Q^*$, then $\delta((q,p),\chi)=((q,\delta_\Tt(p,a,\nChld)),\id)$,
	\item	if $\delta_\Cc(q,\chi)\in\Sigma\times Q^*$, then $\delta((q,\symb{b}[U_1][\dots][U_k]),\chi)=(b,(q,U_1),\dots,(q,U_k))$,
	\item	if $\delta_\Cc(q,\chi)=(a,q_1,\dots,q_\nChld)\in\Sigma\times Q^*$ and $i\in\{1,\dots,\nChld\}$, then $\delta((q,\xVar_{i,p}),\chi)=((q_i,p),\id)$, and
	\item	all other transitions are irrelevant, and can be defined arbitrarily.
	\end{itemize}\smallskip

	It is easy to prove by coinduction that if $\Cc$ is fully convergent from some configuration $(q,s)$,
	then, for every state $p\in P$, $\Cc_\Tt$ generates $\Tt_p(T_{q,s})$ from $((q,p),s)$, where $T_{q,s}$ is the tree generated by $\Cc$ from $(q,s)$.
	Indeed, because $\Cc$ is fully convergent from $(q,s)$, for some $(q',t$) we have $(q,s)\to_\Cc^*(q',t)$ and $\delta(q',\topS(t))=(a,q_1,\dots,q_\nChld)\in\Sigma\times Q^*$.
	In such a situation $((q,p),s)\to_{\Cc_\Tt}^*((q',p),t)$ (where we use transitions of the first kind).
	From $((q',p),t)$ the CPS $\Cc_\Tt$ uses a transition of the second kind,
	and then starts generating the tree $\delta_\Tt(p,a,\nChld)$ until a variable is reached (using transitions of the third kind).
	When a variable $\xVar_{i,p'}$ is reached, $\Cc_\Tt$ enters the configuration $((q_i,p'),t)$ (a transition of the fourth kind),
	which, by the assumption of coinduction, means that it continues by generating the tree $\Tt_{p'}(T_{q_i,t})$, where $T_{q_i,t}$ is the tree generated by $\Cc$ from $(q_i,t)$.

	In particular we have that $\Cc_\Tt$ generates $\Tt(T)$.
	At the end we translate $\Cc_\Tt$ to a scheme $\Gg_\Tt$ generating the same tree, using again Fact~\ref{fact:scheme-equiv-cps}.

\end{document}